%% file: QML_Edge_AI.tex
\documentclass[journal,comsoc]{IEEEtran}

\usepackage{amsbsy,amsmath,amssymb,graphicx,dsfont,mathrsfs,multirow,amsthm}
\usepackage{algpseudocode}
\usepackage{subcaption,float}

\usepackage[colorlinks=true,
            linkcolor=blue,
            citecolor=blue,
            urlcolor=blue]{hyperref}
            
\usepackage[noabbrev]{cleveref}
\crefrangeformat{equation}{#3(#1)#4--#5(#2)#6}

\usepackage[linesnumbered,lined,boxed,ruled,commentsnumbered]{algorithm2e}
\usepackage{blkarray}
\usepackage{booktabs}
\usepackage{pgfplots} \pgfplotsset{compat=1.18}
\usepackage[T1]{fontenc}
\usepackage[utf8]{inputenc}

\usepackage{pifont}

\usepackage{ragged2e}

\newtheorem{proposition}{Proposition}

\newtheorem{lemma}{Lemma}
\newtheorem{theorem}{\textbf{\textsc{Theorem}}}
\newtheorem{remark}{\emph{Remark}}
\newtheorem{corollary}{Corollary}[theorem]

\DeclareMathAlphabet{\mathpzc}{OT1}{pzc}{m}{it}

\usepackage{color}

\ifCLASSOPTIONcompsoc
\usepackage[nocompress]{cite}
\else
\usepackage{cite}
\fi
\ifCLASSINFOpdf
\else
\fi

\usepackage{setspace}

\usepackage{array}
\newcolumntype{M}[1]{>{\centering\arraybackslash}m{#1}}

\usepackage{tikz}
\usetikzlibrary{decorations.pathreplacing}

\usepackage{xcolor}

\newif\ifshowrevisions
\showrevisionsfalse

\definecolor{newrevcolor}{RGB}{153,0,153}

\definecolor{modifiedrevcolor}{RGB}{190,90,0}
\definecolor{addedrevcolor}{RGB}{0,120,60}

\definecolor{numrevcolor}{RGB}{0,105,145}

\begin{document}
	



\title{Proportional-Fair Resource Allocation and Dual-Threshold Early-Exit Inference for Secure Cooperative Multi-Layer Edge Intelligence}

\author{Thai T. Vu, John Le, Tu N. Nguyen, Jun Shen, Quang Vinh Duong, Ha Nguyen

    \IEEEcompsocitemizethanks{
        \IEEEcompsocthanksitem Thai T. Vu, John Le, Jun Shen, and Quang Vinh Duong are with the School of Computing and Information Technology, University of Wollongong, Wollongong, NSW, Australia (e-mail: tienv@uow.edu.au; johnle@uow.edu.au;
        jshen@uow.edu.au;
        duongquangvinh97@gmail.com).
        \IEEEcompsocthanksitem Tu N. Nguyen is with the Department of Computer Science, Kennesaw State University, Georgia, USA (e-mail: tu.nguyen@kennesaw.edu).
        \IEEEcompsocthanksitem Ha Nguyen is a scientist at CodeZX Software Company Limited, Vietnam (e-mail: ha.nguyen.fzx@gmail.com).
        \IEEEcompsocthanksitem This work has been submitted to the IEEE for possible publication. Copyright may be transferred without notice.
    }
}

\IEEEtitleabstractindextext{
\begin{abstract}
\input{0__abstract}		
\end{abstract}

\begin{IEEEkeywords}
    Wireless edge intelligence, fair resource allocation, dual-threshold inference, early-exit CNN, secure offloading, mixed-integer conic optimization.
\end{IEEEkeywords}
}

\maketitle

\IEEEdisplaynotcompsoctitleabstractindextext
\IEEEpeerreviewmaketitle

\input{1__introduction}

\input{3__system_and_problem}

\input{4__proposed_solution}

\input{5__numerical_results}

\input{6__conclusion}

\input{7__appendices}

\bibliographystyle{IEEEtran}
\bibliography{QML_Edge_AI_Refs,related_work_extra}

\end{document}

%% file: 0__abstract.tex
This paper proposes FREDI (Fair Resource Allocation for Edge Dual-Threshold Inference), a secure wireless edge-intelligence framework for event-triggered inference in a cooperative user equipment (UE)--edge server (ES)--cloud system. Each UE performs early-exit convolutional neural network (CNN) screening using dual confidence thresholds, while critical events are securely offloaded to an edge server for detailed classification. We formulate a proportionally-fair utility maximization problem that jointly optimizes UE--ES association, wireless and processing resources, and confidence thresholds. FREDI decomposes the problem into proportional-fair resource allocation and dual-threshold inference optimization. We prove that the detected-critical event set is set-monotone non-increasing in both thresholds, and exploit the finite empirical confidence domain for exact threshold optimization. An empirical resource--utility response envelope yields a computable global suboptimality bound and a sufficient condition for global optimality. By pre-eliminating infeasible UE--ES pairs and exactly projecting out bandwidth and transmit-power variables, the resource-allocation subproblem is reduced to a mixed-integer exponential-cone program solvable to the certified global optimality within a prescribed gap. Numerical results with early-exit MobileNetV2 and ShuffleNetV2 demonstrate near-perfect UE fairness with aggregate utility close to a Sum-Utility benchmark, reveal security-induced resource fragmentation, and demonstrate the Stage-A scalability from 6 to 144 UEs with median solving time below 0.1~s in the tested configurations.

%% file: 1__introduction.tex
\section{Introduction}
\label{sec:introduction}

The growing demand for latency-sensitive intelligent services has accelerated the implementation of \emph{wireless edge intelligence}, in which communication and nearby computing resources are jointly orchestrated for distributed inferences~\cite{letaief2021edge,mendez2022edge}. In a three-layer deployment, user equipment (UEs) observes events close to their sources, edge servers (ESs) provide more capable inference, and a cloud server (CS) periodically updates the models. This hierarchy reduces reliance on remote-cloud inference without requiring resource-constrained UEs to execute the complete learning pipeline.

Early-exit neural networks (EENNs) are well-established for adaptive inference, where intermediate classifiers allow sufficiently confident samples to terminate before the final layer, thereby reducing computation~\cite{teerapittayanon2016branchynet,huang2018multiscale,rahmath2024early}. In particular, confidence-threshold exit policies are common in literature. More generally, two- or multiple-threshold decision rules that partition a confidence domain into decision and rejection/undecided regions have been widely used well before the emergence of edge AI~\cite{chow1970optimum,fumera2000reject,bartlett2008classification}. In this paper, we study how early-exit convolutional neural networks (CNNs) and dual-threshold decision rules would interact with shared wireless and processing resources in a secure and cooperative multi-layer edge system.

Recent studies have increasingly coupled deep neural network (DNN) inference with edge-resource management. Liu \emph{et al.}~\cite{liu2023resource} optimize multi-user communication and computation with batching and early exits in a single-ES system, while Kim and Lee~\cite{kim2024early} jointly optimize edge resources and DNN splitting in an edge--cloud setting. Zheng \emph{et al.}~\cite{zheng2025joint} further consider multi-terminal, multi-base station (BS) semantic transmission with communication/computation allocation. Liu \emph{et al.}~\cite{liu2026joint} jointly optimize model partitioning, exit selection, association, bandwidth, and computation in multi-server mobile edge computing (MEC). 
Zhang \emph{et al.}~\cite{zhang2025joint} and Yuan \emph{et al.}~\cite{yuan2025era} further couple service placement and model splitting with resource allocation. These studies demonstrate the value of co-designing DNN inference and edge-resource allocation, but none jointly considers proportional fairness, security-constrained association, and per-UE dual-threshold screening.

Most closely related, Zhou \emph{et al.}~\cite{zhou2025communication} use dual-threshold multi-exit inference for event-triggered offloading between a single device and an edge server. Our FREDI (Fair Resource Alloca-
tion for Edge Dual-Threshold Inference) framework differs primarily at the network level: multiple UEs compete for communication and processing resources across multiple ESs, coupling threshold selection with UE--ES association, bandwidth, transmit power, ES capacity, security eligibility, and proportional fairness. This turns device-level inference and offloading control into network-wide resource orchestration with combinatorial association.

Fairness and security have largely been studied separately. Xu \emph{et al.}~\cite{xu2022distributed} use proportional fairness for distributed DNN-inference assignment and load balancing, while~\cite{vu2024energy} studies energy-based proportional fairness for task offloading and resource allocation in cooperative edge computing, without considering early-exit inference. Security-aware offloading couples execution decisions with protection requirements and resource use~\cite{xiao2021authentication,peng2024scof}. Hybrid quantum--classical models have also been explored for classification and distributed learning~\cite{abbas2021power,ren2025toward,long2025hybrid,fan2023hybrid}. In FREDI, lightweight binary early-exit screening is performed at UEs, whereas heavier hybrid CNN–QNN (quantum neural network) inference~\cite{fan2023hybrid} is performed at ESs.

Table~\ref{tab:related_work_comparison} compares FREDI with representative studies. A checkmark denotes explicit treatment, $\circ$ for partial or related treatment, and ``--'' indicating an absence as a principal component. The comparison highlights differences in system scope and analytical treatment rather than a comparison of the novelty in standard early-exit or threshold mechanisms.

\begin{table*}[!t]
\centering
\caption{Comparison with representative edge-inference and resource-allocation studies.
Columns denote: confidence-based early exiting; two independently chosen thresholds;
more than one device; more than one server; joint optimization of communication \emph{and}
server-side compute; an explicit proportional-fair objective; association restricted by a
security or trust level; a proved structural property of the threshold rule; and a reported
runtime study over increasing network size. Entries record the aspects each work treats as a
principal component; a dash does not imply the aspect is unattainable within that framework.}
\label{tab:related_work_comparison}
\scriptsize
\setlength{\tabcolsep}{3.2pt}
\renewcommand{\arraystretch}{1.18}
\begin{tabular}{lccccccccc}
\toprule
\textbf{Study} &
\shortstack{\textbf{Early}\\\textbf{exit}} &
\shortstack{\textbf{Dual}\\\textbf{threshold}} &
\shortstack{\textbf{Multi-}\\\textbf{UE/device}} &
\shortstack{\textbf{Multi-}\\\textbf{ES/server}} &
\shortstack{\textbf{Joint network}\\\textbf{resource alloc.}} &
\shortstack{\textbf{PF}} &
\shortstack{\textbf{Security}\\\textbf{constraints}} &
\shortstack{\textbf{Threshold}\\\textbf{structure/proof}} &
\shortstack{\textbf{Network}\\\textbf{scalability}} \\
\midrule
Liu \emph{et al.}, JSAC'23~\cite{liu2023resource}
& $\checkmark$ & -- & $\checkmark$ & -- & $\checkmark$ & -- & -- & -- & -- \\
Xu \emph{et al.}, IoTJ'23~\cite{xu2022distributed}
& -- & -- & $\circ$ & $\checkmark$ & $\circ$ & $\checkmark$ & -- & -- & -- \\
Kim and Lee, IoTJ'24~\cite{kim2024early}
& $\checkmark$ & -- & $\circ$ & $\circ$ & $\checkmark$ & -- & -- & -- & -- \\
Zheng \emph{et al.}, TWC'25~\cite{zheng2025joint}
& $\checkmark$ & -- & $\checkmark$ & $\checkmark$ & $\checkmark$ & -- & -- & -- & -- \\
Zhou \emph{et al.}, TCOM'26~\cite{zhou2025communication}
& $\checkmark$ & $\checkmark$ & -- & -- & $\circ$ & -- & -- & $\circ$ & -- \\
Liu \emph{et al.}, CJE'26~\cite{liu2026joint}
& $\checkmark$ & -- & $\checkmark$ & $\checkmark$ & $\checkmark$ & -- & -- & -- & -- \\
\textbf{FREDI (this work)}
& $\checkmark$ & $\checkmark$ & $\checkmark$ & $\checkmark$ & $\checkmark$ & $\checkmark$ & $\checkmark$ & $\checkmark$ & $\checkmark$ \\
\bottomrule
\end{tabular}
\end{table*}


To address these gaps, we develop \textbf{FREDI (Fair Resource Allocation for Edge Dual-Threshold Inference)}. Its novelty lies not in early-exit or dual-threshold inference alone, but in their network-scale formulation, structural analysis, and provably controlled optimization within a secure cooperative multi-layer edge-AI system. The main contributions are:

\begin{itemize}
    \item \textbf{Network-scale cooperative formulation:} We formulate a proportional-fair multi-UE, multi-ES problem that jointly optimizes per-UE dual thresholds, security-constrained UE--ES association, wireless resources, and shared ES processing capacity, thereby capturing competition for heterogeneous communication and computing resources.

    \item \textbf{Structural analysis and optimization guarantees:} We prove that the detected-critical event set is set-monotone non-increasing in both thresholds, enabling exact finite-domain threshold optimization with directional pruning. We further eliminate infeasible UE--ES pairs and exactly project out bandwidth and power variables, yielding a mixed-integer exponential-cone resource-allocation problem with a certified optimality gap. A computable resource--utility response envelope bounds the end-to-end decomposition loss and gives a sufficient condition for global optimality.
    
    \item \textbf{Fairness, security, and scalability evaluation:} Experiments with early-exit MobileNetV2 and ShuffleNetV2 quantify the dual- versus coupled single-parameter threshold trade-off, fairness--utility trade-off, and security-induced resource fragmentation, while testing scalability from $N=6$ to $N=144$ UEs. FREDI achieves near-uniform UE utility with practical computation over the tested network sizes.
\end{itemize}

The rest of this paper is organized as follows. Section~\ref{sec:model_and_problem} presents the system model, performance metrics, and joint optimization problem. Section~\ref{sec:solution:TMC} details FREDI, including its two-stage decomposition and optimality analysis. Section~\ref{sec:performanceevaluation:TMC} reports experimental results on thresholding, fairness, security constraints, and scalability. Section~\ref{sec:conclusion:TMC} concludes the paper.

%% file: 3__system_and_problem.tex
\section{System Model and Problem Formulation}
\label{sec:model_and_problem}

\input{3_1_system_model}
\input{3_2_performance_metrics}
\input{3_3_problem_formulation}

%% file: 3_1_system_model.tex
\subsection{System Model}
\label{sec:system_model}

We consider a cooperative three-layer Edge AI system comprising $N$ local user devices (i.e., surveillance cameras), referred to as UEs and denoted by $\mathbb{N}=\{1,2,\ldots,N\}$; $M$ edge servers (ESs), denoted by $\mathbb{M}=\{1,2,\ldots,M\}$; and one cloud server (CS). The three-layer system entities are shown on the left of Fig.~\ref{fig:integrated_system_inference}, while the right side summarizes the cooperative inference and periodic model-update workflow.
During each scheduling period, UE~$n$ captures a sequence of independent events (or image), denoted by $\Phi_n=\{I_{n1},I_{n2},\ldots,I_{n|\Phi_n|}\}$, where $I_{nk}$ denotes the $k$th event and $|\Phi_n|$ is the number of events. 
Each UE~$n\in\mathbb{N}$ employs a lightweight CNN with dual-threshold early exits to classify locally observed events as \textit{normal} or \textit{critical}. Events detected as normal remain at the UE. For an event detected as critical, the UE securely offloads its extracted feature representation to one associated ES. The selected ES executes a hybrid CNN--QNN model for detailed multi-class classification and returns the resulting label or response action.
The CS collects training information or model updates from the ESs, updates the classical CNN, feature-adapter, and QNN parameters, and distributes the updated hybrid models back to the ESs. The model parameters are treated as fixed during each scheduling and inference period considered by the optimization.

Let $\mathbb{S}=\{1,\ldots,S\}$ denote the set of security levels assigned to UEs, ESs, and the CS, where a larger index represents a stronger security level. Each UE~$n$ serves one application with a required security level $s_n^{\mathrm{UE}}\in\mathbb{S}$; consequently, all events generated by that UE share the same security requirement. Their critical-event features can therefore be offloaded only to an ES~$j$ satisfying $s_j^{\mathrm{ES}}\geq s_n^{\mathrm{UE}}$. We assume that the CS satisfies the required security and privacy protections.

\begin{figure}[!t]
    \centering
    \includegraphics[width=0.98\columnwidth]
    {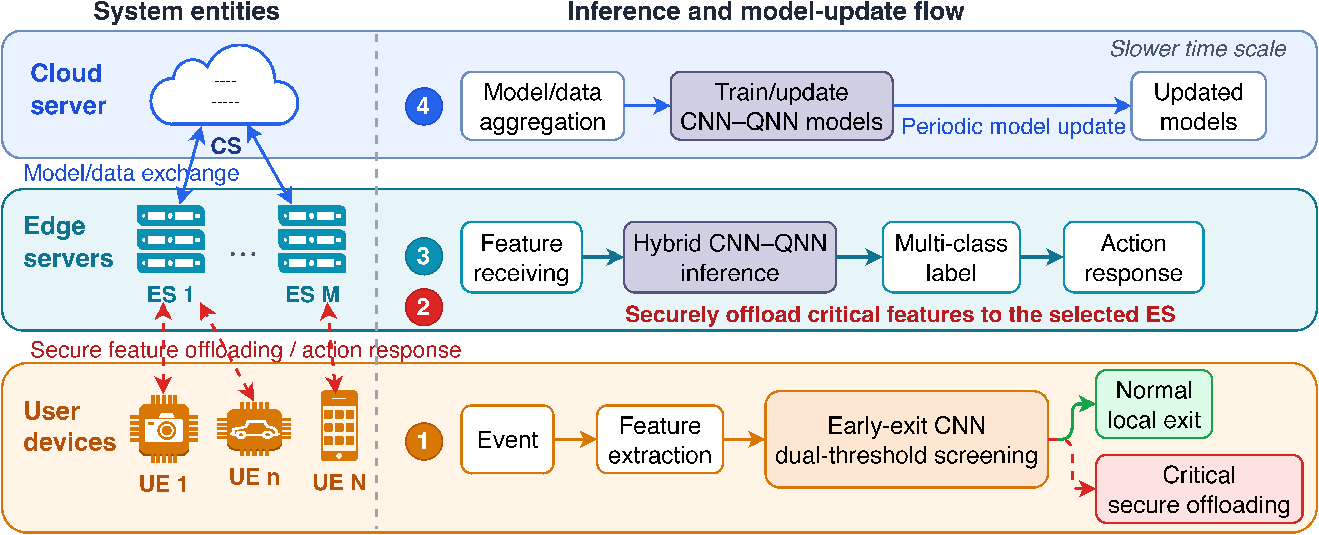}
    \caption{Three-layer Edge-AI system entities (left) and cooperative inference and periodic model-update workflow (right).}
    \label{fig:integrated_system_inference}
\end{figure}

\subsubsection{Early-Exit Inference at UEs}

The UE-side CNN consists of convolutional feature extractors, pooling or subsampling layers, nonlinear activations, and lightweight classification heads~\cite{patil2020convolutional}. Because events vary in classification difficulty, we adopt early-exit inference~\cite{teerapittayanon2016branchynet,huang2018multiscale,rahmath2024early} with a two-threshold confidence rule~\cite{chow1970optimum,fumera2000reject}. At each exit, an event is classified as critical if its confidence reaches the upper threshold, as normal if it reaches the lower threshold, and otherwise proceeds to a deeper exit. The UE-specific thresholds are denoted by $(\alpha_n^l,\alpha_n^u)$, with $0\leq\alpha_n^l\leq\alpha_n^u\leq1$. Fig.~\ref{fig:CNN_with_early_exit} illustrates this policy. Our focus is the coupling of threshold-controlled inference loads with secure multi-UE/multi-ES resource orchestration, rather than the threshold rule itself.

\begin{figure}[!t]
	\centering
	\includegraphics[width=0.75\columnwidth]{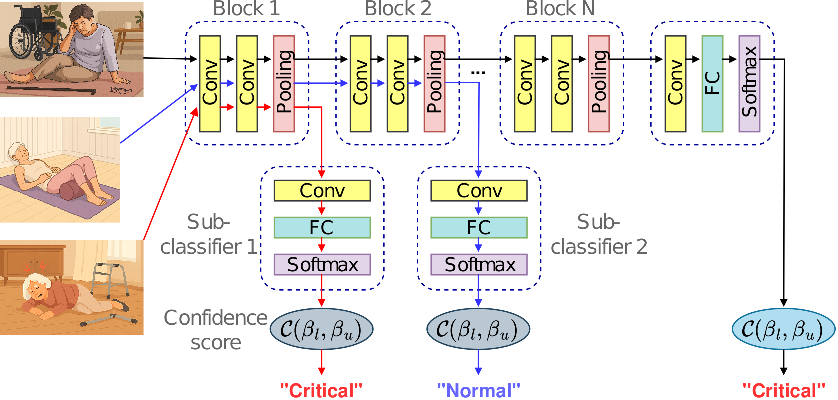}	
	\caption{CNN with dual-threshold early exits.}
	\label{fig:CNN_with_early_exit}
\end{figure}


\textbf{Confidence Score.}
Let $L$ be the number of layers of the lightweight CNN at UE~$n$.
Given the extracted feature of event~$I_{nk}\in\Phi_n$ as input, the CNN generates a confidence score at each layer to assess whether the event is critical. The confidence score $C_{nq}^{(k)}$ at layer~$q$ for event~$I_{nk}$ is defined as

\begin{equation}
    C_{nq}^{(k)}
    =
    \frac{e^{f_{nq}^{(k),\mathrm{critical}}}}
    {e^{f_{nq}^{(k),\mathrm{critical}}}+e^{f_{nq}^{(k),\mathrm{normal}}}},
    \label{eq:confidence_score}
\end{equation}
where $f_{nq}^{(k),\mathrm{critical}}$ and $f_{nq}^{(k),\mathrm{normal}}$ are the critical- and normal-class logits, respectively, at layer~$q$ for event~$I_{nk}$.

\textbf{Early-Exit Event Detection.}
The CNN classifies each event as \textit{critical} ($1$) or \textit{normal} ($0$) as early as possible. For $\alpha_n^l<\alpha_n^u$, processing stops at the first exit~$q$ whose confidence leaves the undecided interval $\mathcal I_n\triangleq(\alpha_n^l,\alpha_n^u)$: event~$I_{nk}$ is declared normal if $C_{nq}^{(k)}\leq\alpha_n^l$ and critical if $C_{nq}^{(k)}\geq\alpha_n^u$. If all exits remain within the interval, the event is conservatively declared normal at the final exit. When $\alpha_n^l=\alpha_n^u$, the policy reduces to a single decision at the first exit, with equality assigned to the normal class.

Accordingly, the predicted label $\hat{y}_{nk}\in\{1,0\}$ is

\begin{equation}
\hat{y}_{nk} =
\begin{cases}
0, \;
    \exists\,q\leq L:\ C_{nq}^{(k)}\leq\alpha_n^l,\ 
    C_{nt}^{(k)}\in\mathcal I_n,\ \forall t<q, \\

1, \;
    \exists\,q\leq L:\ C_{nq}^{(k)}\geq\alpha_n^u>\alpha_n^l,\ 
    C_{nt}^{(k)}\in\mathcal I_n,\ \forall t<q, \\

0, \;
    C_{nt}^{(k)}\in\mathcal I_n,\ \forall t\leq L, \\

\mathds{1}
    \left\{C_{n1}^{(k)}>\alpha_n^l\right\}, 
    \alpha_n^l=\alpha_n^u.
\end{cases}
\label{eq:label_function}
\end{equation}
where $\mathds{1}\{\cdot\}$ is the indicator function. The four cases are mutually exclusive: when $\alpha_n^l<\alpha_n^u$ the last case is inactive, and when $\alpha_n^l=\alpha_n^u$ the interval $\mathcal I_n$ is empty, so only the first exit is examined and confidence exactly equal to the common threshold is assigned to the normal class.

\subsubsection{Secure Communication Between UEs and ESs}

We employ FDMA for UE--ES uplink communication. Let $\boldsymbol{b}=\{b_{nj}\}\in\mathbb{R}^{N\times M}$ and $\boldsymbol{p}=\{p_{nj}\}\in\mathbb{R}^{N\times M}$ denote the bandwidth and transmit power allocated from UE~$n$ to ES~$j$, with per-link limits $b^{\max}$ and $p^{\max}$. Each UE is associated with one ES, represented by $\mathbf{x}=\{x_{nj}\}\in\{0,1\}^{N\times M}$, where $x_{nj}=1$ denotes association between UE~$n$ and ES~$j$.

Let $D_n$ denote the representative size, in bits, of the feature data offloaded by UE~$n$ for one event. The uplink rate from UE~$n$ to ES~$j$ is
\begin{equation}
r_{nj}(b_{nj},p_{nj})
=
b_{nj}\log_2\!\left(
1+\frac{g_{nj}p_{nj}}{\sigma_{nj}^{2}b_{nj}}
\right),
\label{eq:uplink_rate}
\end{equation}
where $g_{nj}$ and $\sigma_{nj}^{2}$ are the channel gain and noise power spectral density of the legitimate link, respectively.

To model physical-layer confidentiality, we consider an eavesdropper that attempts to intercept transmissions from UE~$n$. Its achievable rate on the bandwidth allocated to link $(n,j)$ is
\begin{equation}
r_{nj}^{\mathrm{EV}}(b_{nj},p_{nj})
=
b_{nj}\log_2\!\left(
1+\frac{g_{n}^{\mathrm{EV}}p_{nj}}
{(\sigma_{n}^{\mathrm{EV}})^{2}b_{nj}}
\right),
\label{eq:eavesdropper_rate}
\end{equation}
where $g_n^{\mathrm{EV}}$ and $(\sigma_n^{\mathrm{EV}})^2$ denote the corresponding eavesdropper-channel gain and noise power spectral density. The secure transmission rate~\cite{wyner1975wire} is
\begin{equation}
r_{nj}^{\mathrm{se}}(b_{nj},p_{nj})
=
\left[
r_{nj}(b_{nj},p_{nj})
-r_{nj}^{\mathrm{EV}}(b_{nj},p_{nj})
\right]^+,
\label{eq:secure_rate}
\end{equation}
where $[z]^+=\max\{0,z\}$. Consistent with the FDMA model, transmissions are orthogonal in frequency, so~\eqref{eq:uplink_rate} and~\eqref{eq:eavesdropper_rate} contain no inter-user or inter-cell interference term; the eavesdropper channel statistics $g_n^{\mathrm{EV}}$ and $(\sigma_n^{\mathrm{EV}})^2$ are assumed known at the scheduler, which is the standard worst-case assumption in secrecy-rate resource allocation and yields a conservative feasible set. The selected secure rate of UE~$n$ is
\begin{equation}
r_n^{\mathrm{se}}
(\boldsymbol{b}_n,\boldsymbol{p}_n,\boldsymbol{x}_n)
=
\sum_{j=1}^{M}x_{nj}
 r_{nj}^{\mathrm{se}}(b_{nj},p_{nj}).
\label{eq:selected_secure_rate}
\end{equation}
For a positive selected secure rate, the secure uplink transmission delay is
\begin{equation}
t_n^{\mathrm{tx}}
(\boldsymbol{b}_n,\boldsymbol{p}_n,\boldsymbol{x}_n)
=
\frac{D_n}
{r_n^{\mathrm{se}}
(\boldsymbol{b}_n,\boldsymbol{p}_n,\boldsymbol{x}_n)}.
\label{eq:transmission_time}
\end{equation}

Let $t_n^{\mathrm{r}}$ denote the maximum allowable uplink transmission delay of UE~$n$.

The real-time latency model considers only the secure uplink transmission delay from UEs to ESs. The execution and queueing delays associated with hybrid-model inference at the ESs, as well as the comparatively slower CS model-update delay, are beyond the scope of this work.

\subsubsection{Hybrid CNN--QNN Multi-Class Inference at ESs}

ES~$j$ is characterized by the tuple $(B_j,W_j,s_j^{\mathrm{ES}})$, where $B_j$ is its available UE-uplink bandwidth, $W_j$ is its aggregate event-processing capacity over the considered scheduling period, and $s_j^{\mathrm{ES}}\in\mathbb{S}$ is its security level. Each ES hosts a hybrid CNN--QNN classifier that processes the received representation and produces the final multi-class prediction. We measure $W_j$ in event-processing units, with one unit corresponding to the workload required to process one offloaded event through the complete hybrid inference pipeline, assuming approximately homogeneous per-event workloads.

If UE~$n$ is associated with ES~$j$, the system allocates uplink bandwidth $b_{nj}$ and reserves event-processing capacity $w_{nj}$ at ES~$j$. Let $w^{\max}$ denote the maximum processing capacity that an ES can reserve for one UE during a scheduling period.

When an event is detected as critical at UE~$n$, its feature representation is offloaded to the selected ES for multi-class classification and the corresponding response. Rather than explicitly modeling ES execution time, processing feasibility is captured by the aggregate capacity $W_j$ and per-UE allocation $w_{nj}$; the optimization is therefore agnostic to the ES classifier. We adopt the hybrid CNN--QNN instantiation as a forward-looking design point for when variational quantum classifiers become practical, and it is not evaluated in this paper.

\subsubsection{Periodic Model Aggregation and Updating at the CS}
Periodically, the ESs transmit approved training information, model parameters, or feature summaries to the CS over secure backhaul links, as illustrated by the upper model-update path on the right of Fig.~\ref{fig:integrated_system_inference}. The CS aggregates this information, updates the hybrid CNN--QNN model, and redistributes the updated model to the ESs for subsequent deployment. Cloud-training delay and backhaul-resource allocation are outside the optimization scope.

%% file: 3_2_performance_metrics.tex

\subsection{Performance Metrics}
\label{sec:performance_metrics}

\subsubsection{Categories of Input Events and Output Results}

For each event~$I_{nk}\in\Phi_n$, UE~$n$ uses its lightweight CNN to classify the event as \textit{critical} or \textit{normal}. Let $y_{nk}$ and $\hat{y}_{nk}$ denote the ground-truth and predicted labels of event~$I_{nk}$, respectively. The event is correctly classified if $\hat{y}_{nk}=y_{nk}$; the four possible outcomes are the usual true/false positive/negative combinations, summarized as event sets in Table~\ref{tab:events_categories}.

Over a given time period, UE~$n$ captures the event set $\Phi_n=\{I_{n1},\ldots,I_{n|\Phi_n|}\}$. 
Let $\Phi_n^{\mathrm{N}}$ and $\Phi_n^{\mathrm{P}}$ denote the sets of normal and critical events at UE~$n$, respectively. Let $\hat{\Phi}_n^{\mathrm{TP}}$, $\hat{\Phi}_n^{\mathrm{FP}}$, $\hat{\Phi}_n^{\mathrm{TN}}$, and $\hat{\Phi}_n^{\mathrm{FN}}$ denote the sets of true-positive (TP), false-positive (FP), true-negative (TN), and false-negative (FN) events, respectively, as summarized in Table~\ref{tab:events_categories}.

During the considered scheduling period, every event predicted as critical is offloaded to an ES. We therefore define the offloaded load of UE~$n$ as
\begin{equation}
L_n(\alpha_n^l,\alpha_n^u)
\triangleq
|\hat{\Phi}_n^{\text{TP}}|
+
|\hat{\Phi}_n^{\text{FP}}|.
\label{eq:offloaded_load}
\end{equation}
Using the association and processing-allocation variables, define the event-processing capacity allocated to UE~$n$ as
\begin{equation}
R_n
\triangleq
\sum_{j=1}^{M}x_{nj}w_{nj}.
\label{eq:allocated_processing_capacity}
\end{equation}
Because $w_{nj}$ is measured in event-processing units over the same period, the allocated capacity must cover the offloaded load but need not exceed the number of events generated by the UE. These processing QoS requirements are specified in Problem~$\mathbf{P}_0$.
\begin{table}[!t]
\caption{Categories of events at UE~$n$.}
\label{tab:events_categories}
\centering
\footnotesize
\setlength{\tabcolsep}{7.0pt}
\renewcommand{\arraystretch}{1.05}
\begin{tabular}{cll}
\toprule
\textbf{Set} & \textbf{Definition} & \textbf{Type} \\
\midrule
$\Phi_n$                     & All events at UE~$n$               & Input  \\
$\Phi_n^{\mathrm{N}}$        & Normal events at UE~$n$            & Input  \\
$\Phi_n^{\mathrm{P}}$        & Critical events at UE~$n$          & Input  \\
$\hat{\Phi}_n^{\mathrm{TP}}$ & Correctly detected critical events & Output \\
$\hat{\Phi}_n^{\mathrm{FP}}$ & Normal events detected as critical & Output \\
$\hat{\Phi}_n^{\mathrm{TN}}$ & Correctly detected normal events   & Output \\
$\hat{\Phi}_n^{\mathrm{FN}}$ & Critical events detected as normal & Output \\
\bottomrule
\end{tabular}
\end{table}

We then define the performance metrics that are influenced by the use of dual confidence thresholds.

\subsubsection{User Utility Functions}
Because correctly detecting critical events is the primary objective of the system, we define the utility of UE~$n$, denoted by $\mathcal{U}_n(\alpha_n^l,\alpha_n^u)$, as the proportion of critical events correctly classified as critical. This quantity is the true-positive rate (TPR). From Table~\ref{tab:events_categories} we have
\begin{equation}
    \mathcal{U}_{n}(\alpha_n^l, \alpha_n^u) = \frac{|\hat{\Phi}_n^{\text{TP}}|}{|\Phi_n^{\text{P}}|}.
\label{eq:user_utility}
\end{equation}

The communication and edge-processing resources required by UE~$n$ increase with its offloaded load $L_n(\alpha_n^l,\alpha_n^u)$. We therefore optimize the UE utilities subject to the wireless and processing resources available at the ESs.

%% file: 3_3_problem_formulation.tex
\subsection{Problem Formulation}
\label{sec:problem_formulation}

Events arrive sequentially at the UEs, and only those predicted as critical are offloaded to ESs for hybrid CNN--QNN multi-class classification. Over one scheduling period, the joint design contains two coupled decision groups that motivate FREDI: (i) fair UE--ES association and wireless/edge resource allocation through $(\mathbf{x},\boldsymbol{b},\boldsymbol{p},\boldsymbol{w})$, and (ii) early-exit CNN inference control through the lower and upper confidence thresholds $(\boldsymbol{\alpha}^{l},\boldsymbol{\alpha}^{u})$. The CNN and QNN model parameters distributed by the CS are fixed during this period; periodic cloud training is therefore outside $\mathbf{P}_0$.

Maximizing aggregate utility alone may allocate few communication and processing resources to UEs whose utility gain per unit resource is comparatively small. FREDI therefore associates \emph{fairness with the resource-allocation layer}: we adopt weighted proportional fairness across UEs, while the dual thresholds remain inference-control variables. We assume $\mathcal{U}_n(\alpha_n^l,\alpha_n^u)>0$ for every UE~$n$. 
Weighted proportional fairness is induced by maximizing the standard network utility $\sum_{n=1}^{N}\rho_n\ln(\mathcal{U}_n)$, where $\rho_n>0$ is the fairness weight of UE~$n$~\cite{kelly1998rate}.

The end-to-end problem nevertheless optimizes both groups jointly: the dual confidence thresholds $\boldsymbol{\alpha}^{l}$ and $\boldsymbol{\alpha}^{u}$ determine early-exit CNN decisions and offloaded load, whereas $\mathbf{x}$, $\boldsymbol{b}$, $\boldsymbol{p}$, and $\boldsymbol{w}$ determine fair secure access to wireless and ES processing resources. The resulting problem is

\begin{subequations}
\label{eq:prob_0}

\begin{align}
\mathbf{P}_0:\quad
\underset{
\boldsymbol{\alpha}^{l},\boldsymbol{\alpha}^{u},
\mathbf{x},\boldsymbol{b},\boldsymbol{p},\boldsymbol{w}
}{
\operatorname{maximize}
}
\quad
\sum_{n=1}^{N}\rho_n
\ln\!\left(
\mathcal{U}_n(\alpha_n^l,\alpha_n^u)
\right)
\tag{\ref{eq:prob_0}}
\end{align}

\noindent\textnormal{subject to}

\begin{align}
&0\leq b_{nj}\leq b^{\max},
&&\forall (n,j)\in\mathbb{N}\times\mathbb{M},
\label{eq:cons_band_user_0}
\\
&0\leq p_{nj}\leq p^{\max},
&&\forall (n,j)\in\mathbb{N}\times\mathbb{M},
\label{eq:cons_pow_user_0}
\\
&0\leq w_{nj}\leq w^{\max},
&&\forall (n,j)\in\mathbb{N}\times\mathbb{M},
\label{eq:cons_processing_user_0}
\\
&\sum_{n=1}^{N}x_{nj}b_{nj}\leq B_j,
&&\forall j\in\mathbb{M},
\label{eq:cons_band_edge_0}
\\
&\sum_{n=1}^{N}x_{nj}w_{nj}\leq W_j,
&&\forall j\in\mathbb{M},
\label{eq:cons_comp_edge_0}
\\
&\sum_{j=1}^{M}x_{nj}s_j^{\mathrm{ES}}
\geq s_n^{\mathrm{UE}},
&&\forall n\in\mathbb{N},
\label{eq:cons_security_user_0}
\\
&t_n^{\mathrm{tx}}
\left(
\boldsymbol{b}_n,
\boldsymbol{p}_n,
\boldsymbol{x}_n
\right)
\leq t_n^{\mathrm{r}},
&&\forall n\in\mathbb{N},
\label{eq:cons_delay_user_0}
\\
&R_n
\leq |\Phi_n|,
&&\forall n\in\mathbb{N},
\label{eq:cons_comp_demand_user_0}
\\
&\sum_{j=1}^{M}x_{nj}=1,
&&\forall n\in\mathbb{N},
\label{eq:cons_conn_user_0}
\\
&x_{nj}\in\{0,1\},
&&\forall (n,j)\in\mathbb{N}\times\mathbb{M},
\label{eq:cons_binary_variables_0}
\\
&L_n(\alpha_n^l,\alpha_n^u)
\leq
R_n,
&&\forall n\in\mathbb{N},
\label{eq:cons_comp_user_0}
\\
&0\leq\alpha_n^l\leq\alpha_n^u\leq1,
&&\forall n\in\mathbb{N}.
\label{eq:cons_dual_thresholds_variables_0}
\end{align}

\end{subequations}

Constraints~\eqref{eq:cons_band_user_0}--\eqref{eq:cons_processing_user_0} specify the per-link bandwidth, transmission-power, and event-processing-reservation limits, respectively. Constraint~\eqref{eq:cons_conn_user_0} associates each UE with exactly one ES. Constraints~\eqref{eq:cons_band_edge_0} and~\eqref{eq:cons_comp_edge_0} enforce the aggregate bandwidth and event-processing capacities of each ES. Constraint~\eqref{eq:cons_security_user_0} ensures that a UE is associated only with an ES whose security level is sufficiently strong, while~\eqref{eq:cons_delay_user_0} bounds its secure uplink delay. Finally, constraints~\eqref{eq:cons_comp_user_0} and~\eqref{eq:cons_comp_demand_user_0} enforce the processing QoS chain $L_n(\alpha_n^l,\alpha_n^u)\leq R_n\leq|\Phi_n|$.

The coupling between the resource-allocation variables and the dual-threshold load constraint~\eqref{eq:cons_comp_user_0} is the central structure exploited by FREDI in Section~\ref{sec:solution:TMC}.

%% file: 4__proposed_solution.tex
\section{The Proposed FREDI Framework}
\label{sec:solution:TMC}

In this section, we introduce the proposed method, Fair Resource Allocation for Edge Dual-Threshold Inference (FREDI). Its two stages deliberately assign proportional fairness to multi-UE wireless/edge resource allocation and dual-threshold optimization to early-exit CNN inference.

\input{4_1_problem_characterization}

\input{4_2_decomposition}
\input{4_3_subproblem_a}
\input{4_4_subproblem_b}

%% file: 4_1_problem_characterization.tex
\subsection{Problem Characterization}

Problem~$\mathbf{P}_0$ is a nonlinear mixed-integer problem: it contains binary association
variables $x_{nj}$, empirical and generally non-differentiable utilities
$\mathcal{U}_n(\alpha_n^l,\alpha_n^u)$, and products coupling the association decisions to the
continuous allocations $(b_{nj},p_{nj},w_{nj})$. FREDI exploits this structure through two
complementary stages. Stage~A performs proportional-fair UE--ES association and wireless/edge resource allocation and is transformed exactly into a mixed-integer exponential-cone problem. Stage~B then performs dual-threshold inference optimization for each UE over the finite set of observed confidence scores under the processing capacity allocated by Stage~A. Hence, fairness is governed by Stage~A, whereas inference and offloading decisions are controlled by Stage~B. The following analysis distinguishes the exactness of each stage from the end-to-end loss induced by their sequential composition.

\begin{lemma}[Set-monotonicity of dual-threshold early exiting]
\label{lem:set_decreasing_monotonic}
For fixed CNN confidence traces, let $\hat{\Phi}^{\mathrm{critical}}(\alpha^l,\alpha^u)$ denote the events classified as critical under thresholds $0\leq\alpha^l\leq\alpha^u\leq1$. For any two admissible threshold pairs satisfying $\alpha^{l\prime}\geq\alpha^l$ and $\alpha^{u\prime}\geq\alpha^u$,
\begin{equation}
\hat{\Phi}^{\mathrm{critical}}(\alpha^{l\prime},\alpha^{u\prime})
\subseteq
\hat{\Phi}^{\mathrm{critical}}(\alpha^l,\alpha^u).
\label{eq:critical_set_monotonicity}
\end{equation}
Consequently, $|\hat{\Phi}^{\mathrm{critical}}|$ is monotonically non-increasing in either threshold over the admissible domain.
\end{lemma}

\emph{Proof:} The proof is shown in Appendix~\ref{sec:set_decreasing_monotonic}.

%% file: 4_2_decomposition.tex
\subsection{Fair Resource Allocation for Edge Dual-Threshold Inference}

The empirical utility depends directly on the inference thresholds $(\alpha_n^l,\alpha_n^u)$, whereas the wireless and processing allocations determine which threshold pairs are feasible through the offloaded-load constraint~\eqref{eq:cons_comp_user_0}. Recall from~\eqref{eq:allocated_processing_capacity} that $R_n=\sum_{j=1}^{M}x_{nj}w_{nj}$ is the processing capacity allocated to UE~$n$.

To expose the resource--inference coupling without assuming that the reserved capacity is fully utilized, let $\xi>0$ denote a generic processing-capacity level and define the optimal Stage-B response of UE~$n$ as
\begin{equation}
V_n(\xi)
\triangleq
\max_{\substack{0\leq\alpha_n^l\leq\alpha_n^u\leq1\\
L_n(\alpha_n^l,\alpha_n^u)\leq \xi}}
\mathcal{U}_n(\alpha_n^l,\alpha_n^u).
\label{eq:stage_b_value_function}
\end{equation}
Only capacity levels for which the feasible set contains a positive-utility threshold pair are considered, consistently with the logarithmic domain assumed in~$\mathbf P_0$. Define the corresponding empirical response factor as
\begin{equation}
\kappa_n(\xi)
\triangleq
\frac{V_n(\xi)}{\xi/|\Phi_n|}.
\label{eq:resource_utility_response_factor}
\end{equation}
Consequently, the optimized end-to-end objective associated with a feasible resource allocation admits the exact decomposition
\begin{equation}
\sum_{n=1}^{N}\rho_n\ln V_n(R_n)
=
\underbrace{\sum_{n=1}^{N}\rho_n
\ln\!\left(\frac{R_n}{|\Phi_n|}\right)}_{A(\mathbf R)}
+
\underbrace{\sum_{n=1}^{N}\rho_n\ln\kappa_n(R_n)}_{K(\mathbf R)}.
\label{eq:exact_objective_decomposition}
\end{equation}
FREDI optimizes the processing-coverage term $A(\mathbf R)$ in Stage~A and then realizes the best empirical inference response $V_n(R_n)$ in Stage~B. Hence, the only end-to-end approximation introduced by the decomposition is that Stage~A does not optimize the response term $K(\mathbf R)$; this term characterizes the departure of the empirical inference response from exact proportionality to the normalized reserved capacity.

\textbf{FREDI Stage A: Fair UE--ES association and resource allocation.}
The variables $(\mathbf{x},\boldsymbol{b},\boldsymbol{p},\boldsymbol{w})$ are optimized to allocate secure communication resources and event-processing capacity proportionally across UEs. The Stage-A problem is
\begin{subequations}
\label{eq:prob_a}
\begin{align}
\mathbf{P}_A:\quad
\underset{\mathbf{x},\boldsymbol{b},\boldsymbol{p},\boldsymbol{w}}
{\operatorname{maximize}}
\quad&
\sum_{n=1}^{N}\rho_n
\ln\!\left(
\frac{\sum_{j=1}^{M}x_{nj}w_{nj}}
{|\Phi_n|}
\right)
\tag{\ref{eq:prob_a}}
\end{align}
\noindent\textnormal{subject to constraints~\eqref{eq:cons_band_user_0}--\eqref{eq:cons_delay_user_0} and~\eqref{eq:cons_comp_demand_user_0}--\eqref{eq:cons_binary_variables_0}.}
\end{subequations}

Let $(\mathbf{x}^{*},\boldsymbol{b}^{*},\boldsymbol{p}^{*},\boldsymbol{w}^{*})$ denote a globally optimal Stage-A allocation and $R_n^{*}=\sum_{j=1}^{M}x_{nj}^{*}w_{nj}^{*}$. Since $|\Phi_n|$ is fixed within an optimization instance, $\sum_n\rho_n\ln(R_n/|\Phi_n|)=\sum_n\rho_n\ln(R_n)-\sum_n\rho_n\ln|\Phi_n|$; therefore, the normalized and unnormalized Stage-A objectives have identical maximizers. The normalization is retained to interpret $R_n/|\Phi_n|$ as the fraction of the UE workload for which processing capacity is reserved.

\textbf{FREDI Stage B: Dual-threshold inference optimization.}
With the Stage-A allocation fixed, the thresholds are selected by
\begin{align}
\mathbf{P}_B:\quad
&\underset{\boldsymbol{\alpha}^{l},\boldsymbol{\alpha}^{u}}
{\operatorname{maximize}}
&&\sum_{n=1}^{N}\rho_n
\ln\!\left(\mathcal{U}_n(\alpha_n^l,\alpha_n^u)\right)
\label{eq:prob_b}
\end{align}
\noindent\textnormal{subject to the threshold-domain constraint~\eqref{eq:cons_dual_thresholds_variables_0} and the per-UE capacity constraints}
\begin{equation}
L_n(\alpha_n^l,\alpha_n^u)\leq R_n^{*},
\qquad \forall n\in\mathbb{N}.
\label{eq:cons_comp_user_sub_b}
\end{equation}
Problem~$\mathbf P_B$ is separable across UEs and attains $V_n(R_n^*)$ for every UE when each per-UE empirical search is solved exactly.

Let $F_0^{\star}$ denote the globally optimal objective value of the empirical problem~$\mathbf P_0$, and define $F_0^{\mathrm{FREDI}}\triangleq\sum_n\rho_n\ln V_n(R_n^*)$ as the objective obtained by globally solving~$\mathbf P_A$ followed by the exact Stage-B searches.

\begin{theorem}[End-to-end global performance bound]
\label{theo:fredi_global_bound}
Suppose that the Stage-A solution admits a positive-utility Stage-B solution for every UE. Let $\mathcal Q_n^{+}$ denote the finite set of positive-utility threshold candidates for UE~$n$, with $L_n(\alpha_n^l,\alpha_n^u)>0$, and define the computable response envelope
\begin{equation}
\overline{\kappa}_n
\triangleq
\max_{(\alpha_n^l,\alpha_n^u)\in\mathcal Q_n^{+}}
\frac{|\Phi_n|\mathcal U_n(\alpha_n^l,\alpha_n^u)}
{L_n(\alpha_n^l,\alpha_n^u)}.
\label{eq:kappa_upper_envelope}
\end{equation}
Thus, $\overline{\kappa}_n$ is obtained during the same finite candidate enumeration used to solve Stage~B and requires no solution of~$\mathbf P_0$.
Then
\begin{equation}
0\leq
F_0^{\star}-F_0^{\mathrm{FREDI}}
\leq
\sum_{n=1}^{N}\rho_n
\ln\!\left(
\frac{\overline{\kappa}_n}
{\kappa_n(R_n^*)}
\right).
\label{eq:fredi_global_suboptimality_bound}
\end{equation}
\end{theorem}

\emph{Proof:} The proof is shown in Appendix~\ref{sec:global_performance_bound}.

\begin{corollary}[Conditional global optimality]
\label{cor:fredi_global_optimality}
If $\kappa_n(\xi)$ is constant over the positive-utility capacity domain of every UE~$n$, then FREDI globally solves the empirical problem~$\mathbf P_0$.
\end{corollary}

\begin{proof}
Under the stated condition, $\overline{\kappa}_n=\kappa_n(R_n^*)$ for every UE, so the upper bound in~\eqref{eq:fredi_global_suboptimality_bound} is zero.
\end{proof}

The factor $\kappa_n(\xi)=|\Phi_n|V_n(\xi)/\xi$ measures the optimized inference utility obtained per unit of normalized reserved capacity. It therefore reflects both the predictive yield of the selected threshold pair and any unused capacity caused by the discrete empirical operating points. If the Stage-B load is smaller than $R_n^*$, the unused reservation can be released without changing the thresholds or the objective of~$\mathbf P_0$; this release does not retroactively change the Stage-A allocation decision.

\begin{remark}[Solver-certified bound]
If the mixed-integer solver terminates before proving the exact Stage-A optimum, let $\widehat{\mathbf R}$ denote its incumbent processing allocation, $A^{\mathrm{inc}}=A(\widehat{\mathbf R})$ its incumbent objective, and $A^{\mathrm{ub}}$ its certified upper bound. Then
\begin{equation}
F_0^{\star}
\leq
A^{\mathrm{ub}}
+\sum_{n=1}^{N}\rho_n\ln\overline{\kappa}_n.
\label{eq:solver_certified_p0_upper_bound}
\end{equation}
After exact Stage-B optimization under $\widehat{\mathbf R}$, the resulting feasible objective $F_0^{\mathrm{FREDI}}(\widehat{\mathbf R})$ satisfies
\begin{equation}
F_0^{\star}-F_0^{\mathrm{FREDI}}(\widehat{\mathbf R})
\leq
\underbrace{A^{\mathrm{ub}}-A^{\mathrm{inc}}}_{\text{Stage-A solver gap}}
+\underbrace{\sum_{n=1}^{N}\rho_n
\ln\!\left(\frac{\overline{\kappa}_n}{\kappa_n(\widehat R_n)}\right)}_{\text{decomposition bound}}.
\label{eq:solver_certified_end_to_end_gap}
\end{equation}
The second term is evaluated from the Stage-B candidate sets and the realized FREDI operating points; no iterative exchange between the two stages is required.
\end{remark}

The FREDI workflow is illustrated in Fig.~\ref{fig:proposed_optimization_workflow}. Stage~A allocates secure wireless and ES processing resources proportionally across UEs, whereas Stage~B optimizes the dual inference thresholds under the resulting capacities. The $N$ threshold-selection problems in Stage~B are separable and can be solved in parallel using Algorithm~\ref{alg:threshold_selection}. Theorem~\ref{theo:fredi_global_bound} bounds the end-to-end loss relative to the empirical global optimum, and Corollary~\ref{cor:fredi_global_optimality} identifies the proportional-response condition under which the bound vanishes. The detailed solutions of the two stages are presented in Sections~\ref{sec:solve_subproblem_a} and~\ref{sec:solve_subproblem_b}, respectively.

\begin{figure}[!t]
    \centering
    \includegraphics[width=0.90\columnwidth]{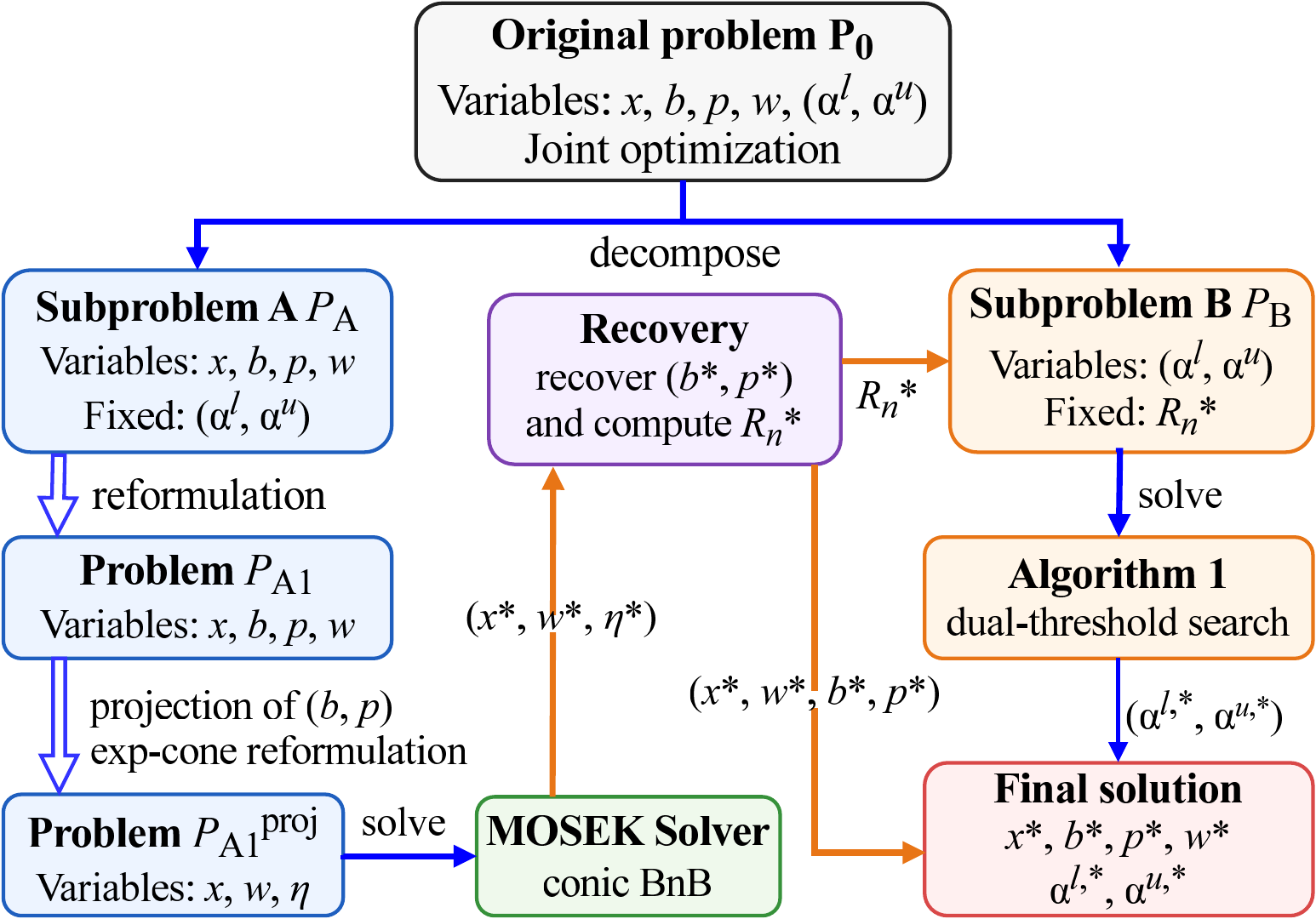}
    \caption{FREDI algorithm workflow.}
    \label{fig:proposed_optimization_workflow}
\end{figure}

%% file: 4_3_subproblem_a.tex
\subsection{FREDI Stage A: Fair Resource Allocation}
\label{sec:solve_subproblem_a}

FREDI Stage~A implements proportional fairness through the UE--ES association and wireless/edge resource-allocation problem $\mathbf{P}_A$. Problem~$\mathbf{P}_A$ contains products between binary association variables and continuous resource variables. Moreover, when there is no connection ($x_{nj} = 0$) between UE~$n$ and ES~$j$, the resource allocation variables can be set to $0$. Thus, constraints~\eqref{eq:cons_band_user_0}, \eqref{eq:cons_pow_user_0}, and~\eqref{eq:cons_processing_user_0} equivalently become
\begin{align}
0\leq b_{nj}&\leq b^{\max}x_{nj}, \qquad \forall (n,j)\in\mathbb{N}\times\mathbb{M}
\label{eq:link_bandwidth}
\\
0\leq p_{nj}&\leq p^{\max}x_{nj},  \qquad \forall (n,j)\in\mathbb{N}\times\mathbb{M}
\label{eq:link_power}
\\
0\leq w_{nj}&\leq w^{\max}x_{nj},  \qquad \forall (n,j)\in\mathbb{N}\times\mathbb{M}
\label{eq:link_processing}
\end{align}
Once these constraints are imposed, the products $x_{nj}b_{nj}$ and $x_{nj}w_{nj}$ are unnecessary in the aggregate capacity constraints. Thus, constraints~\eqref{eq:cons_band_edge_0}, \eqref{eq:cons_comp_edge_0}, and~\eqref{eq:cons_comp_demand_user_0} reduce, respectively, to
\begin{align}
&\sum_{n=1}^{N}b_{nj}\leq B_j,
&&\forall j\in\mathbb{M},
\label{eq:cons_band_edge_a1}
\\
&\sum_{n=1}^{N}w_{nj}\leq W_j,
&&\forall j\in\mathbb{M},
\label{eq:cons_comp_edge_a1}
\\
&\sum_{j=1}^{M}w_{nj}\leq|\Phi_n|,
&&\forall n\in\mathbb{N}.
\label{eq:cons_processing_demand_a1}
\end{align}

For the discrete security-level requirement, define
\begin{equation}
a_{nj}=
\begin{cases}
1, & s_j^{\mathrm{ES}}\geq s_n^{\mathrm{UE}},\\
0, & \text{otherwise},
\end{cases}
\qquad
\forall(n,j)\in\mathbb{N}\times\mathbb{M}.
\label{eq:security_compatibility}
\end{equation}
The original security constraint~\eqref{eq:cons_security_user_0} can then be replaced by
\begin{equation}
x_{nj}\leq a_{nj},
\qquad
\forall(n,j)\in\mathbb{N}\times\mathbb{M}.
\label{eq:cons_security_compatibility}
\end{equation}

\subsubsection{Secure-Rate and Delay Reformulation}

Define the normalized legitimate and eavesdropper channel gains as
\begin{equation}
\gamma_{nj}\triangleq\frac{g_{nj}}{\sigma_{nj}^{2}},
\qquad
\gamma_n^{\mathrm{EV}}
\triangleq
\frac{g_n^{\mathrm{EV}}}{(\sigma_n^{\mathrm{EV}})^{2}}.
\label{eq:normalized_channel_gains}
\end{equation}
For a link that can provide positive secrecy, we require
$\gamma_{nj}>\gamma_n^{\mathrm{EV}}$. The secure rate can then be written without the positive-part operator as
\begin{equation}
\begin{split}
r_{nj}^{\mathrm{se}}(b_{nj},p_{nj})
={}&b_{nj}\log_2\!\left(1+\frac{\gamma_{nj}p_{nj}}{b_{nj}}\right)\\
&-b_{nj}\log_2\!\left(1+\frac{\gamma_n^{\mathrm{EV}}p_{nj}}{b_{nj}}\right),
\end{split}
\label{eq:secure_rate_reformulated}
\end{equation}
with the continuous extension $r_{nj}^{\mathrm{se}}(0,p)=0$ for $p\geq0$. If $x_{nj}=1$, the communication-delay requirement is
\begin{equation}
\frac{D_n}{r_{nj}^{\mathrm{se}}(b_{nj},p_{nj})}
\leq t_n^{\mathrm{r}}.
\label{eq:selected_link_delay}
\end{equation}
Using the binary association variable, this condition is equivalently represented for every link by
\begin{equation}
r_{nj}^{\mathrm{se}}(b_{nj},p_{nj})
\geq
\bar r_n x_{nj},
\qquad
\bar r_n\triangleq\frac{D_n}{t_n^{\mathrm{r}}},
\label{eq:minimum_secure_rate}
\end{equation}
for all $(n,j)\in\mathbb{N}\times\mathbb{M}$. When $x_{nj}=0$, the linking constraints set $b_{nj}=p_{nj}=0$, and~\eqref{eq:minimum_secure_rate} reduces to $0\geq0$. When $x_{nj}=1$, it is equivalent to~\eqref{eq:selected_link_delay}.

To characterize the curvature of the secure-rate constraint, define
\begin{equation}
q_{nj}(z)
=
\log_2(1+\gamma_{nj}z)
-
\log_2(1+\gamma_n^{\mathrm{EV}}z).
\label{eq:secure_rate_auxiliary}
\end{equation}
For $\gamma_{nj}>\gamma_n^{\mathrm{EV}}$, $q_{nj}(z)$ is concave for $z\geq0$. This property is used below to establish the mixed-integer convex structure of $\mathbf{P}_{A1}$. 

Because $|\Phi_n|$ is constant for every UE within the scheduling period, maximizing $\sum_n\rho_n\ln(R_n/|\Phi_n|)$ is equivalent to maximizing $\sum_n\rho_n\ln(R_n)$. After applying the linking constraints, $R_n=\sum_j w_{nj}$, so we use the latter form to obtain a simpler conic representation.

The resulting resource-allocation problem is
\begin{subequations}
\label{eq:prob_a1}
\begin{align}
\mathbf{P}_{A1}:\quad
\underset{
\mathbf{x},\boldsymbol{b},\boldsymbol{p},\boldsymbol{w}
}{
\operatorname{maximize}
}
\quad
&\sum_{n=1}^{N}\rho_n
\ln\!\left(
\sum_{j=1}^{M}w_{nj}
\right)
\tag{\ref{eq:prob_a1}}
\end{align}

\noindent\textnormal{subject to~\eqref{eq:link_bandwidth}--\eqref{eq:link_processing}, \eqref{eq:cons_band_edge_a1}--\eqref{eq:cons_processing_demand_a1}, \eqref{eq:cons_security_compatibility}, \eqref{eq:minimum_secure_rate}, \eqref{eq:cons_conn_user_0}, and~\eqref{eq:cons_binary_variables_0}.}
\end{subequations}

\subsubsection{Feasible-Pair Preprocessing}

A UE--ES pair is individually feasible only if it satisfies the security-level requirement and can support the required secure transmission rate under the per-link resource limits. Define
\begin{equation}
\mathbb{M}_n^{\mathrm{f}}
\triangleq
\left\{
 j\in\mathbb{M}\,\middle|\,
 a_{nj}=1,
 \ \gamma_{nj}>\gamma_n^{\mathrm{EV}},
 \ r_{nj}^{\mathrm{se}}(b^{\max},p^{\max})\geq\bar r_n
\right\}.
\label{eq:feasible_es_set}
\end{equation}
All infeasible association variables are fixed according to
\begin{equation}
x_{nj}=0,
\qquad
\forall n\in\mathbb{N},\; j\notin\mathbb{M}_n^{\mathrm{f}}.
\label{eq:infeasible_pair_fixing}
\end{equation}
If $\mathbb{M}_n^{\mathrm{f}}=\varnothing$ for any UE~$n$, then $\mathbf{P}_{A1}$ is infeasible.

After this fixing, $\mathbf{P}_{A1}$ is a \emph{mixed-integer convex} problem whose only
source of nonconvexity is the binary association. Indeed, every remaining pair satisfies
$\gamma_{nj}>\gamma_n^{\mathrm{EV}}$, so
$q_{nj}''(z)=\frac{1}{\ln 2}[(\gamma_n^{\mathrm{EV}})^2(1+\gamma_n^{\mathrm{EV}}z)^{-2}
-\gamma_{nj}^2(1+\gamma_{nj}z)^{-2}]<0$ for $z\geq0$ and $q_{nj}$ is concave; its closed
perspective $b_{nj}q_{nj}(p_{nj}/b_{nj})=r_{nj}^{\mathrm{se}}(b_{nj},p_{nj})$ is therefore
jointly concave in $(b_{nj},p_{nj})$~\cite[Sec.~3.2.6]{boyd2004convex}, making
$\bar r_nx_{nj}-r_{nj}^{\mathrm{se}}(b_{nj},p_{nj})\leq0$ convex. All other constraints of the
continuous relaxation are affine, and the objective is concave because $\rho_n>0$ and
$\ln(\cdot)$ is concave and increasing.

For every feasible pair, define the minimum bandwidth required at the maximum transmission power:
\begin{equation}
\underline b_{nj}
\triangleq
\min_{0\leq b\leq b^{\max}}
\left\{
 b\,\middle|\,
 r_{nj}^{\mathrm{se}}(b,p^{\max})\geq\bar r_n
\right\}.
\label{eq:min_required_bandwidth}
\end{equation}
The one-dimensional value $\underline b_{nj}$ can be computed efficiently by bisection because $r_{nj}^{\mathrm{se}}(b,p^{\max})$ is continuous and non-decreasing in $b$ over the feasible interval. It is also bounded: $b\,q_{nj}(p/b)\to p(\gamma_{nj}-\gamma_n^{\mathrm{EV}})/\ln 2$ as $b\to\infty$, so a link admits a feasible allocation only if $\bar r_n<p^{\max}(\gamma_{nj}-\gamma_n^{\mathrm{EV}})/\ln 2$, regardless of the available bandwidth. Security-induced infeasibility is therefore a property of the secrecy advantage rather than of spectrum scarcity alone.

\subsubsection{Projection of Communication Variables}

The power variables do not appear in the objective, and no aggregate power constraint is imposed. This structure permits the communication variables to be eliminated exactly rather than handled by a custom branch-and-bound procedure.

\begin{lemma}
\label{lem:exact_communication_projection}
Consider a binary association matrix $\mathbf{x}$ satisfying $\sum_{j=1}^{M}x_{nj}=1$ and $x_{nj}=0$ for all $j\notin\mathbb{M}_n^{\mathrm{f}}$. There exist bandwidth and power allocations $(\boldsymbol{b},\boldsymbol{p})$ satisfying~\eqref{eq:link_bandwidth}, \eqref{eq:link_power}, \eqref{eq:cons_band_edge_a1}, and~\eqref{eq:minimum_secure_rate} if and only if
\begin{equation}
\sum_{n=1}^{N}\underline b_{nj}x_{nj}\leq B_j,
\qquad
\forall j\in\mathbb{M}.
\label{eq:projected_bandwidth_constraint}
\end{equation}
Whenever~\eqref{eq:projected_bandwidth_constraint} holds, one feasible recovery is
\begin{equation}
b_{nj}=\underline b_{nj}x_{nj},
\qquad
p_{nj}=p^{\max}x_{nj}.
\label{eq:communication_recovery}
\end{equation}
\end{lemma}

\emph{Proof:} See Appendix~\ref{sec:proof_exact_communication_projection}.

The recovery in~\eqref{eq:communication_recovery} selects one representative
allocation among possibly many that share the same projected point; it excludes no
feasible association or processing decision.

\begin{remark}
The exact projection relies on the absence of a transmission-energy objective or an aggregate power budget. If either feature is introduced, setting every selected link to $p^{\max}$ may no longer be optimal, and the variables $(\boldsymbol{b},\boldsymbol{p})$ must be retained in the optimization.
\end{remark}

\subsubsection{Mixed-Integer Exponential-Cone Reformulation}

Introduce an auxiliary variable $\eta_n$ for each UE and impose
\begin{equation}
\eta_n\leq\ln\!\left(\sum_{j=1}^{M}w_{nj}\right).
\label{eq:log_auxiliary_constraint}
\end{equation}
Using the primal exponential cone
\begin{equation}
\mathcal{K}_{\exp}
\triangleq
\operatorname{cl}\left\{
(u,v,z)\in\mathbb{R}^{3}\,\middle|\,
 v>0,\ u\geq v\exp(z/v)
\right\},
\end{equation}
constraint~\eqref{eq:log_auxiliary_constraint} is equivalently represented as
\begin{equation}
\left(
\sum_{j=1}^{M}w_{nj},\ 1,\ \eta_n
\right)
\in\mathcal{K}_{\exp}.
\label{eq:log_exponential_cone}
\end{equation}

Because $\rho_n>0$ and the objective maximizes $\sum_n\rho_n\eta_n$,
\eqref{eq:log_exponential_cone} is active at the optimum, so introducing $\eta_n$
linearizes the objective without changing it.

After applying Lemma~\ref{lem:exact_communication_projection}, Subproblem~A is reformulated as
\begin{subequations}
\label{eq:prob_a1_projected}
\begin{align}
\mathbf{P}_{A1}^{\mathrm{proj}}:\quad
\underset{\mathbf{x},\boldsymbol{w},\boldsymbol{\eta}}
{\operatorname{maximize}}
\quad &
\sum_{n=1}^{N}\rho_n\eta_n
\tag{\ref{eq:prob_a1_projected}}
\end{align}

\noindent\textnormal{subject to~\eqref{eq:log_exponential_cone}, \eqref{eq:cons_conn_user_0}, \eqref{eq:projected_bandwidth_constraint}, \eqref{eq:cons_comp_edge_a1}, \eqref{eq:link_processing}, \eqref{eq:cons_processing_demand_a1}, \eqref{eq:infeasible_pair_fixing}, and~\eqref{eq:cons_binary_variables_0}.}
\end{subequations}

Problem~$\mathbf{P}_{A1}^{\mathrm{proj}}$ is accordingly a mixed-integer
exponential-cone problem: all constraints are affine except for the
exponential-cone memberships, and the only discrete variables are the
binary association variables.

\begin{proposition}
\label{prop:pa1_miecp_equivalence}
Problems~$\mathbf{P}_{A1}$ and $\mathbf{P}_{A1}^{\mathrm{proj}}$ have the same optimal objective value. Furthermore, any global optimizer $(\mathbf{x}^{*},\boldsymbol{w}^{*},\boldsymbol{\eta}^{*})$ of $\mathbf{P}_{A1}^{\mathrm{proj}}$, together with the communication recovery in~\eqref{eq:communication_recovery}, yields a global optimizer of $\mathbf{P}_{A1}$. Conversely, every feasible solution of $\mathbf{P}_{A1}$ maps to a feasible solution of $\mathbf{P}_{A1}^{\mathrm{proj}}$ with the same objective value by setting
\begin{equation}
\eta_n=\ln\!\left(\sum_{j=1}^{M}w_{nj}\right).
\end{equation}
\end{proposition}

\begin{proof}
Lemma~\ref{lem:exact_communication_projection} and $\eta_n=\ln(\sum_jw_{nj})$ map every original feasible solution to an equally valued projected one, so $\operatorname{opt}(\mathbf P_{A1}^{\mathrm{proj}})\geq\operatorname{opt}(\mathbf P_{A1})$. Conversely, the lemma recovers $(\boldsymbol b,\boldsymbol p)$, while $\rho_n>0$ makes~\eqref{eq:log_exponential_cone} tight at optimum, yielding the reverse inequality and equivalence.
\end{proof}

We solve $\mathbf{P}_{A1}^{\mathrm{proj}}$ with the MOSEK mixed-integer conic
solver~\cite{mosek_fusion_2026} and recover
$(\boldsymbol{b}^{*},\boldsymbol{p}^{*})$ via~\eqref{eq:communication_recovery}; the
solver's certified gap therefore applies to $\mathbf{P}_{A1}$ itself.

%% file: 4_4_subproblem_b.tex
\subsection{FREDI Stage B: Dual-Threshold Inference Optimization}
\label{sec:solve_subproblem_b}

FREDI Stage~B acts on the inference layer: with $(\mathbf{x}^{*},\boldsymbol{b}^{*},\boldsymbol{p}^{*},\boldsymbol{w}^{*})$ fixed, Subproblem~$\mathbf{P}_B$ separates across UEs. Because $\rho_n>0$ and $\ln(\cdot)$ is strictly increasing, the threshold pair of UE~$n$ can be obtained from
\begin{align}
\mathbf{SP}_n:\quad
&\underset{\alpha_n^l,\alpha_n^u}{\operatorname{maximize}}
&&\mathcal{U}_n(\alpha_n^l,\alpha_n^u)
\\
&\text{subject to}
&&0\leq\alpha_n^l\leq\alpha_n^u\leq1,
\\
&&&L_n(\alpha_n^l,\alpha_n^u)\leq R_n^{*},
\label{eq:cons_comp_user_sub_b_n}
\end{align}
where $R_n^{*}=\sum_{j=1}^{M}x_{nj}^{*}w_{nj}^{*}$. By construction of Subproblem~A, $R_n^{*}\leq|\Phi_n|$; hence every feasible threshold solution satisfies the QoS chain $L_n(\alpha_n^l,\alpha_n^u)\leq R_n^{*}\leq|\Phi_n|$. 

Stage~B therefore controls the detection--offloading behavior of the CNN through the label
rule~\eqref{eq:label_function}, rather than introducing fairness into the inference rule.

\textbf{Role of monotonicity.}
Lemma~\ref{lem:set_decreasing_monotonic} implies that the empirical utility and offloaded load are non-increasing in both thresholds. Hence, an optimal feasible pair lies on the lower boundary of the feasible threshold region, which permits directional pruning during an exact finite search.

\textbf{Exact empirical search.}
The empirical utility and load are piecewise constant because they are obtained by counting samples. Their values can change only when a threshold crosses an observed confidence score. It is therefore sufficient to search threshold pairs drawn from the finite set of observed scores.

\begin{algorithm}[t]
\caption{Dual-Threshold Selection for UE~$n$}
\label{alg:threshold_selection}
\KwIn{Observed confidence-score set $\mathcal{C}_n$; event labels; capacity $R_n^{*}$}
\KwOut{Optimal feasible thresholds $(\alpha_n^{l,*},\alpha_n^{u,*})$ and response envelope $\overline{\kappa}_n$}
Sort the unique values in $\mathcal{C}_n\cup\{0,1\}$
into $c_1<\cdots<c_K$\;
Set $U^{\mathrm{best}}\leftarrow-\infty$\;
Set $\overline{\kappa}_n\leftarrow0$\;
\For{$i\leftarrow1$ \KwTo $K$}{
    \For{$j\leftarrow i$ \KwTo $K$}{
        Set $(\alpha_n^l,\alpha_n^u)\leftarrow(c_i,c_j)$\;
        Evaluate $L_n(\alpha_n^l,\alpha_n^u)$ and
        $\mathcal{U}_n(\alpha_n^l,\alpha_n^u)$\;
        \If{$L_n>0$ and $\mathcal{U}_n>0$}{
            $\overline{\kappa}_n\leftarrow
            \max\{\overline{\kappa}_n,
            |\Phi_n|\mathcal{U}_n/L_n\}$\;
        }
        \If{$L_n(\alpha_n^l,\alpha_n^u)\leq R_n^{*}$ and
        $\mathcal{U}_n>U^{\mathrm{best}}$}{
            Store the threshold pair and update $U^{\mathrm{best}}$\;
        }
    }
}
\Return{$(\alpha_n^{l,*},\alpha_n^{u,*})$ and $\overline{\kappa}_n$}\;
\end{algorithm}

The exhaustive search is globally optimal over the empirical threshold domain. The additional envelope update evaluates~\eqref{eq:kappa_upper_envelope} without changing the selected thresholds. Monotonicity can be used to stop exploring a row or column once all remaining pairs are dominated or infeasible, provided that any pruned candidates cannot increase either the feasible utility or the response envelope.

%% file: 5__numerical_results.tex
\section{Numerical Results}
\label{sec:performanceevaluation:TMC}

We evaluate FREDI from three complementary perspectives. Scenario~S1 isolates the dual-threshold (DT) inference component against a single-threshold (ST) policy; Scenario~S2 evaluates the end-to-end fairness--utility trade-off and validates the two-stage decomposition; and Scenario~S3 examines security-constrained association and Stage-A scalability.

\subsection{Experimental Setup}
\label{subsec:experimental_setup}

\subsubsection{Dataset, Models, and Implementation}

We use a binary cats--dogs corpus with 4,001 cat and 4,006 dog images for training and 1,011 cat and 1,012 dog images for testing~\cite{schubert2018catsdogs}. Images are resized to $224\times224$ and normalized using ImageNet statistics~\cite{deng2009imagenet}. 
UE-side screening uses early-exit MobileNetV2~\cite{sandler2018mobilenetv2} and ShuffleNetV2~\cite{ma2018shufflenet}, with intermediate heads formed by global average pooling and a two-class linear projection. MobileNetV2 and ShuffleNetV2 contain 18 and 9 exits, respectively. 
These lightweight backbones are representative of confidence-based early-exit inference~\cite{teerapittayanon2016branchynet,huang2018multiscale}.

Both models are trained on Google Cloud Platform using four vCPUs, 15~GB RAM, and one NVIDIA T4 GPU for 30 epochs, with batch size 64, learning rate $10^{-3}$, weight decay $10^{-4}$, and equally weighted cross-entropy losses across exits. Final-exit validation performance selects the checkpoint. For S1, exit-wise logits on the complete 2,023-image held-out test set are converted to positive-class softmax confidence traces; all threshold policies therefore operate on identical CNN outputs without retraining.

For S2 and S3, each UE generates $|\Phi_n|=100$ events per scheduling window and MobileNetV2 is used throughout. Unless otherwise stated, $M=3$, $B_j=20$~MHz, $D_n=128$~KiB, and $t_n^{\mathrm r}=100$~ms. Random geometry generates physically consistent legitimate and eavesdropper channel gains but is not itself swept. The projected problem $\mathbf{P}_{A1}^{\mathrm{proj}}$ is implemented using the MOSEK Fusion API~\cite{mosek_fusion_2026}; minimum-bandwidth coefficients $\underline b_{nj}$ are computed by bisection and infeasible UE--ES pairs are removed before model construction. In S3-B, the relative MIP gap is $10^{-4}$ and one warm-up solve is discarded before timing. 
Per-link caps are $b^{\max}=20$~MHz, $p^{\max}=0.2$~W ($23$~dBm) and $w^{\max}=100$ events/window, the noise PSD is $-174$~dBm/Hz. 
UEs are uniformly distributed in a $70$-m triangular serving region under distance based path loss (exponent $4$) and Rayleigh fading. The threshold grid uses $\Delta_\alpha=0.0025$.



\subsubsection{Scenario Design and Compared Methods}

The following defines the compared policies, normalized loads, and statistical protocol.

\textit{Scenario S1: Threshold policies for early-exit screening.}
Scenario~S1 isolates threshold selection from wireless and multi-UE resource-allocation effects. The allocated edge-processing budget is parameterized by the normalized ratio
\begin{equation}
\eta_R = \frac{R^{\star}}{|\Phi|},
\label{eq:normalized_processing_budget}
\end{equation}
with $\eta_R\in\{0.1,0.2,\ldots,1.0\}$. The proposed dual-threshold (DT) policy optimizes $(\alpha^l,\alpha^u)$ over $0\leq\alpha^l\leq\alpha^u\leq1$, whereas the coupled single-parameter (ST) policy uses $\tau\in[0.5,1]$ with $\alpha^l=1-\tau$ and $\alpha^u=\tau$, the symmetric confidence rule used in conventional early-exit inference. Because $1-\tau\leq\tau$ for every admissible $\tau$, the ST family is a one-dimensional subset of the DT family; S1 is therefore an ablation that isolates the value of decoupling the two thresholds, and DT is optimal for ST whenever the two are compared under the same processing budget. The comparison quantifies how much this additional degree of freedom is worth, and where the coupling becomes binding. Both use the same threshold-grid resolution $\Delta_\alpha=0.0025$. Refining to $\Delta_\alpha=0.00125$ at the most sensitive MobileNetV2 point, $\eta_R=0.3$, leaves the selected optimum and reported metrics unchanged. TPR-optimal ties are resolved by lower offloaded load, lower normalized MACs, higher F1 score, and then threshold order.

\textit{Scenario S2: Utility and decomposition validation.}
Scenario~S2 is the main system-level experiment. Each UE stream contains 50 positive and 50 negative held-out events sampled without replacement within a realization, with identical streams and network realization shared across methods. Sweeping $W_j\in\{200,160,120,100,80,60\}$ gives the normalized processing load 
\begin{equation}
\lambda_W
=\frac{\sum_{n=1}^{N}|\Phi_n|}{\sum_{j=1}^{M}W_j}
=\frac{200}{W_j}
\in\{1,1.25,1.67,2,2.5,3.33\}.
\label{eq:normalized_processing_load}
\end{equation}
FREDI is compared with a Sum-Utility benchmark over the same feasible associations and resources, maximizing $\sum_n\mathcal{U}_n$. Since multiple optima may yield different fairness, we report the minimum and maximum Jain's indices. We also validate the decomposition by exactly solving a finite empirical form of $\mathbf{P}_0$ at $\lambda_W\in\{1.25,2,3.33\}$, maximizing $\sum_n\rho_n\ln(\mathcal{U}_n)$ over the same feasible candidates.

\textit{Scenario S3: security constraints and scalability.}
Scenario~S3 stress-tests FREDI Stage~A through two complementary sub-experiments. In S3-A, FREDI is compared with an Unrestricted-PF reference under identical channels and resources. Unrestricted-PF removes only the UE--ES security-eligibility constraint while retaining the same proportional-fair resource allocation. In S3-B, FREDI is evaluated as the number of UEs increases.

In S3-A, seven UE security-demand profiles produce $\kappa_{\mathrm{elig}}\in\{1,0.89,\ldots,0.33\}$, each with 50 paired channel realizations sharing seeds across the two methods. Security-induced infeasibility is retained as an experimental outcome rather than resampled away.

In S3-B, $N\in\{6,12,24,36,60,96,144\}$ with approximately balanced UE security tiers. To separate runtime growth from increasing resource scarcity, the per-ES capacities scale as
\begin{equation}
W_j(N)=\frac{100N}{3},
\qquad
B_j(N)=20\left(\frac{N}{9}\right)\ \text{MHz}.
\label{eq:scalability_capacity_scaling}
\end{equation}
Fifty random realizations are used for each system size.

S1 is deterministic on the fixed test set. S2 and S3 use 50 realizations per operating point. Continuous quantities are reported as sample means with 95\% confidence intervals unless otherwise stated; S3-B additionally reports median solve time, whereas the binary S3-A infeasibility rate uses 95\% Wilson confidence intervals. For S2, the confidence intervals describe variability over the paired event-stream and network realizations drawn from the fixed held-out pool.

\subsubsection{Performance Metrics}

For S1, the primary detection metrics are true-positive rate (TPR), false-positive rate (FPR), and F1 score. 
The computation saving is
\begin{equation}
G_{\mathrm{comp}}=1-\frac{1}{|\Phi|}
\sum_{i\in\Phi}
\frac{\operatorname{MAC}(q_i)}{\operatorname{MAC}(Q)},
\label{eq:computation_saving}
\end{equation}
where $q_i$ is the exit used for event~$i$, $\operatorname{MAC}(q_i)$ is the cumulative MAC count up to that exit, and $Q$ is the final exit.

For a nonnegative vector $\mathbf{z}=(z_1,\ldots,z_N)$,
Jain's index~\cite{jain1984quantitative} is
\begin{equation}
J(\mathbf{z})
=
\frac{\left(\sum_{n=1}^{N}z_n\right)^2}
{N\sum_{n=1}^{N}z_n^2}.
\label{eq:jain_satisfaction_index}
\end{equation}

For S2, let $U_n\equiv\mathcal{U}_n(\alpha_n^l,\alpha_n^u)$ denote the final empirical utility. We report the Jain's indices
\begin{equation}
J_R=J(\{R_n^*/|\Phi_n|\}_{n=1}^{N}),
\qquad
J_U=J(\{U_n\}_{n=1}^{N}),
\label{eq:s3_jain_indices}
\end{equation}
and the aggregate utility $\sum_n U_n$. To characterize the empirical response factor in~\eqref{eq:resource_utility_response_factor} at the realized FREDI allocation, define
\begin{equation}
\kappa_n^*
\equiv\kappa_n(R_n^*)
=\frac{U_n}{R_n^*/|\Phi_n|},
\qquad
\mathrm{CV}_{\kappa}
=\frac{\operatorname{std}(\{\kappa_n^*\})}
{\operatorname{mean}(\{\kappa_n^*\})}.
\label{eq:s3_kappa_cv}
\end{equation}
Finally, the decomposition is compared with direct empirical optimization through
\begin{equation}
F_0^{\mathrm{FREDI}}
=\sum_{n=1}^{N}\rho_n\ln U_n^{\mathrm{FREDI}},
\qquad
F_0^{\mathrm{direct}}
=\sum_{n=1}^{N}\rho_n\ln U_n^{\mathrm{direct}}.
\label{eq:s3_direct_p0_objectives}
\end{equation}
The direct formulation excludes zero-utility empirical operating points because the logarithmic objective is undefined there.

The aggregate service ratio used in S3-A is
\begin{equation}
S
=
\frac{\sum_{n=1}^{N}R_n}
{\sum_{n=1}^{N}|\Phi_n|}.
\label{eq:aggregate_service_ratio}
\end{equation}
For S3-A, define $e_{nj}=1$ when UE~$n$ is security-eligible for ES~$j$, and $e_{nj}=0$ otherwise. The eligible-association ratio is
\begin{equation}
\kappa_{\mathrm{elig}}
=
\frac{1}{NM}
\sum_{n=1}^{N}\sum_{j=1}^{M}e_{nj},
\label{eq:eligible_association_ratio}
\end{equation}
and the infeasible-realization rate is
\begin{equation}
P_{\mathrm{infeas}}
=
\frac{K_{\mathrm{infeas}}}{K}.
\label{eq:infeasible_realization_rate}
\end{equation}
We additionally report the actual feasible-pair ratio after preprocessing and ES processing and bandwidth utilization. For S3-B, we report MOSEK solve time and solver optimality status.

\subsection{Dual- versus Single-Threshold Screening Performance}
\label{subsec:dual_threshold_results}

Fig.~\ref{fig:s1_tpr_computation_saving} and Table~\ref{tab:s1_operating_points}
compare DT and ST under the same allocated edge-processing budget. Both policies use
identical backbones and identical exit-wise confidence traces, so the differences isolate
the effect of decoupling the two thresholds.

The coupling is most damaging under tight budgets. Because ST ties $\alpha^l=1-\tau$ to
$\alpha^u=\tau$, reducing the offloaded load forces $\tau\to1$, which raises the critical-exit
threshold and widens the undecided interval $\mathcal I$ at the same time. For MobileNetV2 at
$\eta_R\in\{0.1,0.2\}$ the load cap admits exactly one ST setting, $\tau=1$: then
$(\alpha^l,\alpha^u)=(0,1)$, $\mathcal I$ spans the whole confidence range, no sample crosses
either threshold, and all are declared normal at the final exit by~\eqref{eq:label_function}.
Table~\ref{tab:s1_operating_points} records the consequence --- a true-positive rate of exactly
zero, and no computation saving either, since nothing exits early. DT decouples the two
decisions and attains TPRs of 0.197 and 0.396 with computation savings of 91.1\% and 67.5\% at
the same budgets. ShuffleNetV2 does not collapse ($\tau=0.9975$ stays feasible), so the effect
depends on how sharply a backbone concentrates its confidence traces; DT still leads it on every
metric there, most visibly in computation saving (55.0\% against 5.6\% at $\eta_R=0.1$).

At the intermediate budget $\eta_R=0.3$, DT improves MobileNetV2 TPR from 0.576 to 0.593
and F1 from 0.729 to 0.742 while raising computation saving from 35.2\% to 58.8\%, at
essentially the same offloaded load (29.0\% to 29.9\%). For ShuffleNetV2 the two policies
select operating points with identical TPR, FPR, F1 and offloaded load, and DT gains only in
computation saving (30.9\% to 35.9\%). At $\eta_R=0.5$ both backbones show small TPR and F1
gains, and the computation benefit becomes model-dependent.

Above $\eta_R=0.5$ the comparison changes character, because ST cannot spend the budget it is
given. Its offloaded set is $\{$events whose confidence reaches $\tau$ at some exit$\}$, which
on the balanced test corpus saturates near half the events as $\tau\to0.5$: Table~\ref{tab:s1_operating_points}
shows the ST load rising only from 49.9\% at $\eta_R=0.5$ to 50.5\% at $\eta_R=0.8$, while DT
reaches 79.7\%. The apparently unfavourable DT entries at $\eta_R=0.8$ --- FPR 0.599 and F1
0.767 against 0.094 and 0.912 for ST --- therefore compare a policy operating at its budget
against one operating far below it. They also reflect the objective rather than the policy class: $\mathbf{SP}_n$
maximizes TPR subject to the load cap and never penalizes false positives, so DT spends
whatever capacity Stage~A reserves. Where precision matters, the same machinery accepts an
F1- or FPR-constrained utility without any change to Stage~A.

The value of the second threshold is therefore not a uniform gain on every classification
metric, but the removal of a structural constraint: it keeps detection alive where the coupled
policy collapses, and it makes the full range of processing budgets reachable.

\begin{figure}[t]
    \centering
    \includegraphics[
        width=0.80\columnwidth
    ]{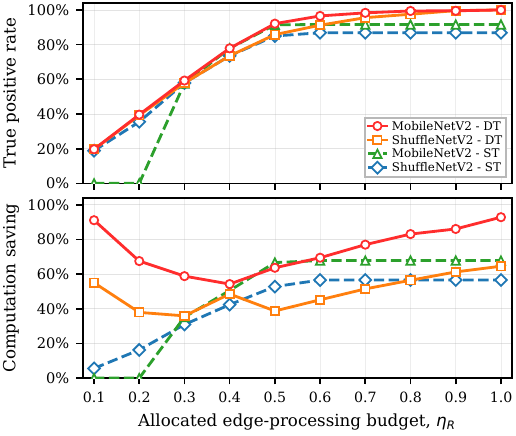}
    \caption{True positive rate (top) and computation saving under the dual-threshold (DT)
    and coupled single-parameter (ST) policies as the allocated edge-processing budget $\eta_R$ increases.}
    \label{fig:s1_tpr_computation_saving}
\end{figure}

\begin{table}[t]
\caption{Representative S1 operating points on the 2,023-image held-out test set.}
\label{tab:s1_operating_points}
\centering
\footnotesize
\setlength{\tabcolsep}{4.0pt}
\renewcommand{\arraystretch}{1.05}

\resizebox{\columnwidth}{!}{%
\begin{tabular}{cccccccc}
\toprule
CNN & Policy & $\eta_R$ & TPR & FPR & F1 & Offloaded load & Comp.\ saving \\
\midrule
MobileNetV2 & DT & 0.1 & \textbf{0.197} & 0.003 & \textbf{0.328} & 10.0\% & \textbf{91.1\%} \\
MobileNetV2 & ST & 0.1 & 0.000 & 0.000 & 0.000 & 0.0\% & 0.0\% \\
MobileNetV2 & DT & 0.2 & \textbf{0.396} & 0.003 & \textbf{0.566} & 20.0\% & \textbf{67.5\%} \\
MobileNetV2 & ST & 0.2 & 0.000 & 0.000 & 0.000 & 0.0\% & 0.0\% \\
MobileNetV2 & DT & 0.3 & \textbf{0.593} & 0.005 & \textbf{0.742} & 29.9\% & \textbf{58.8\%} \\
MobileNetV2 & ST & 0.3 & 0.576 & 0.004 & 0.729 & 29.0\% & 35.2\% \\
MobileNetV2 & DT & 0.8 & \textbf{0.994} & 0.599 & 0.767 & 79.7\% & \textbf{83.1\%} \\
MobileNetV2 & ST & 0.8 & 0.916 & 0.094 & \textbf{0.912} & 50.5\% & 67.8\% \\
\midrule
ShuffleNetV2 & DT & 0.3 & 0.579 & 0.020 & 0.724 & 30.0\% & \textbf{35.9\%} \\
ShuffleNetV2 & ST & 0.3 & 0.579 & 0.020 & 0.724 & 30.0\% & 30.9\% \\
ShuffleNetV2 & DT & 0.8 & \textbf{0.975} & 0.615 & 0.753 & 79.5\% & 56.4\% \\
ShuffleNetV2 & ST & 0.8 & 0.869 & 0.204 & \textbf{0.838} & 53.6\% & \textbf{56.6\%} \\
\bottomrule
\end{tabular}%
}

\vspace{1mm}

\parbox{0.96\columnwidth}{\footnotesize
DT and ST are evaluated using identical exit-wise confidence traces.
Bold entries highlight selected stronger detection or computation results
at the same processing budget.}
\end{table}

\subsection{Utility Performance and Decomposition Validation}
\label{subsec:end_to_end_validation}

The top panel of Fig.~\ref{fig:s2_jain_utility} compares the final-utility fairness of FREDI with the fairness envelope of Sum-Utility as processing contention increases. With sufficient processing capacity ($\lambda_W=1$), both methods attain $J_U=1$. FREDI remains essentially perfectly fair throughout the sweep, with mean $J_U\geq0.9993$. In contrast, the Sum-Utility envelope widens under heavier scarcity: at $\lambda_W=2.5$, its mean lower and upper boundaries are approximately 0.9824 and 0.9953, and at $\lambda_W=3.33$ they become approximately 0.8561 and 0.9851. Hence, maximizing aggregate utility alone can admit substantially less fair optima even though an arbitrary solver tie-break may conceal this behavior.

\begin{figure}[t]
    \centering
    \includegraphics[
        width=0.80\columnwidth
    ]{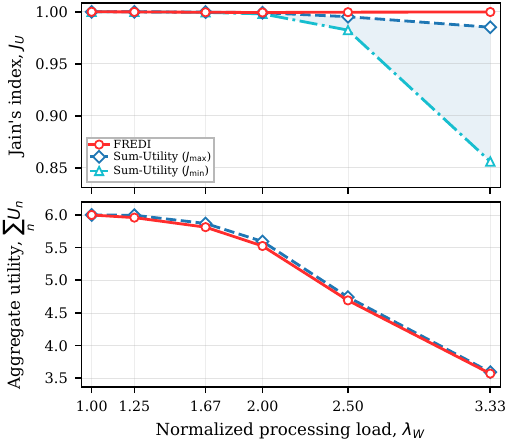}
    \caption{Jain's index $J_U$ (top) and aggregate utility (bottom) for FREDI and the Sum-Utility benchmark as the normalized processing load $\lambda_W$ increases.}
    \label{fig:s2_jain_utility}
\end{figure}

The fairness improvement is obtained with a small aggregate-utility cost. The bottom panel of Fig.~\ref{fig:s2_jain_utility} shows that Sum-Utility, by construction, gives the largest $\sum_n U_n$, while FREDI remains close over the full load range. At $\lambda_W=1$, both attain an aggregate utility of 6. As contention increases, the largest mean difference over the six operating points is 0.0692 at $\lambda_W=2$, where FREDI and Sum-Utility obtain 5.5252 and 5.5944, respectively. Under severe scarcity ($\lambda_W=3.33$), the corresponding values are 3.5712 and 3.5936. Thus, FREDI maintains near-uniform user utility while sacrificing little aggregate detection utility.

Table~\ref{tab:s2_validation} assesses the decomposition from two complementary viewpoints. First, the normalized processing allocation is perfectly balanced at every tested load, $J_R=1$, and the final utility remains extremely close to this allocation-level fairness, with $J_U\geq0.9993$. At the realized FREDI operating points, the mean coefficient of variation of $\kappa_n^*=U_n/(R_n^*/|\Phi_n|)$ never exceeds 2.71\%. The repeated small cross-UE dispersion values show that the realized response residual remains balanced across UEs over the load sweep; they do not by themselves establish the per-UE constancy required by Corollary~\ref{cor:fredi_global_optimality} or evaluate the response envelope $\overline{\kappa}_n$ in Theorem~\ref{theo:fredi_global_bound}.

Second, we compare FREDI with direct finite empirical optimization of the original logarithmic objective. All 150 direct-$\mathbf{P}_0$ validation instances (50 paired realizations at each of three representative loads) are solved to optimal integer status. As required, $F_0^{\mathrm{direct}}\geq F_0^{\mathrm{FREDI}}$ in every case within numerical tolerance. The mean objectives remain close, particularly under severe scarcity, where they are $-3.1138$ and $-3.1008$ for FREDI and direct optimization, respectively. This direct comparison gives a realized empirical assessment of the end-to-end loss, complementing the general upper bound in Theorem~\ref{theo:fredi_global_bound}. Together, the fairness consistency, small $\mathrm{CV}_{\kappa}$, and direct-objective comparison support the FREDI two-stage decomposition over the evaluated operating regime without presuming exact end-to-end equivalence.

\begin{table}[t]
\caption{Empirical validation of the FREDI decomposition and direct finite empirical optimization of $\mathbf{P}_0$.}
\label{tab:s2_validation}
\centering
\footnotesize
\setlength{\tabcolsep}{4.0pt}
\renewcommand{\arraystretch}{1.05}
\begin{tabular}{cccccc}
\toprule
$\lambda_W$ & $J_R$ & $J_U$ & $\mathrm{CV}_{\kappa}$ (\%) & $F_0^{\mathrm{FREDI}}$ & $F_0^{\mathrm{direct}}$ \\
\midrule
1.00 & 1.0 & 1.0000 & 0.00 & -- & -- \\
1.25 & 1.0 & 0.9999 & 1.06 & $-0.0409\pm0.0065$ & $-0.0040\pm0.0023$ \\
1.67 & 1.0 & 0.9995 & 2.27 & -- & -- \\
2.00 & 1.0 & \textbf{0.9993} & \textbf{2.71} & $-0.4969\pm0.0138$ & $-0.4228\pm0.0128$ \\
2.50 & 1.0 & 0.9996 & 2.06 & -- & -- \\
3.33 & 1.0 & 0.9998 & 1.34 & $-3.1138\pm0.0109$ & $-3.1008\pm0.0084$ \\
\bottomrule
\end{tabular}
\vspace{1mm}

\parbox{0.96\columnwidth}{\footnotesize Values are averaged over 50 paired realizations. $J_R=J(\{R_n^*/|\Phi_n|\})$ and $J_U=J(\{U_n\})$. The $F_0$ entries report mean $\pm$ 95\% confidence-interval half-width. Direct-$\mathbf{P}_0$ validation is evaluated at three representative processing loads. Bold entries mark the most stringent observed FREDI fairness-consistency point, where the minimum $J_U$ and maximum $\mathrm{CV}_{\kappa}$ occur.}
\end{table}

\subsection{Impact of Security Constraints and System Scalability}
\label{subsec:security_scalability_results}

\subsubsection{Security-Constrained Association}
\label{subsec:security_constraint_results}

Fig.~\ref{fig:s3_security_pressure} jointly shows the service and feasibility effects of security-induced restrictions on UE--ES association. As shown in the top panel, when $\kappa_{\mathrm{elig}}\geq0.778$, FREDI serves essentially all offered workload and remains close to the eligibility-unconstrained PF resource-allocation upper bound (Unrestricted-PF). Below this region, the service ratio deteriorates rapidly, falling to about 0.661 at $\kappa_{\mathrm{elig}}=0.444$ and to 0.333 at $\kappa_{\mathrm{elig}}=0.333$.

The bottom panel of Fig.~\ref{fig:s3_security_pressure} shows the corresponding infeasible-realization rate of the security-constrained FREDI formulation. All realizations remain feasible for $\kappa_{\mathrm{elig}}\geq0.778$, whereas $P_{\mathrm{infeas}}$ increases to 0.14 at $\kappa_{\mathrm{elig}}=0.444$ and 0.22 at $\kappa_{\mathrm{elig}}=0.333$. Hence, moderate security restrictions can be absorbed without loss of feasibility, whereas severe restrictions simultaneously reduce service capability and increase the probability of an infeasible association/resource-allocation instance.

\begin{figure}[t]
    \centering
    \includegraphics[
        width=0.80\columnwidth
    ]{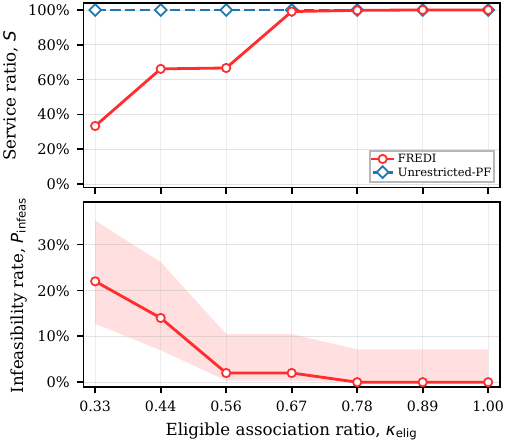}
    \caption{Service ratio (top) and infeasible-realization rate (bottom) of FREDI as the eligible association ratio $\kappa_{\mathrm{elig}}$ increases.}
    \label{fig:s3_security_pressure}
\end{figure}

The mechanism is resource fragmentation rather than a lack of aggregate capacity. The measured feasible-pair ratio closely tracks $\kappa_{\mathrm{elig}}$, confirming that security eligibility dominates the sweep, while the maximum ES bandwidth utilization rises from about 0.30 to 0.64. Under the most restrictive profile, mean processing utilization is only about 0.333 although one eligible ES is fully utilized. Capacity therefore remains stranded at ESs that some UEs are not permitted to access. Moreover, Jain fairness can still be high when service is poor (e.g., $J=1$ at $\kappa_{\mathrm{elig}}=0.333$), showing that fairness and service capability must be interpreted jointly.

\subsubsection{Solver Scalability}
\label{subsec:solver_scalability_results}

Fig.~\ref{fig:s3_scalability} reports the computational scalability of FREDI Stage~A, with communication and processing resources scaled with $N$ to maintain a comparable operating load. The median MOSEK solve time increases from about 2~ms at $N=6$ to 59~ms at $N=144$ and remains below 0.1~s even for the largest tested system. The wider interquartile range at large $N$ reflects realization-dependent mixed-integer search induced by different feasible UE--ES association sets.

All 350 measured instances reach an optimal integer status within the prescribed relative-gap tolerance of $10^{-4}$. Together with preprocessing that removes infeasible UE--ES pairs before model construction, these results show that the FREDI security-constrained resource-allocation formulation remains computationally practical over the investigated system sizes without relaxing solution quality.

\begin{figure}[!t]
    \centering
    \includegraphics[
        width=0.80\columnwidth
    ]{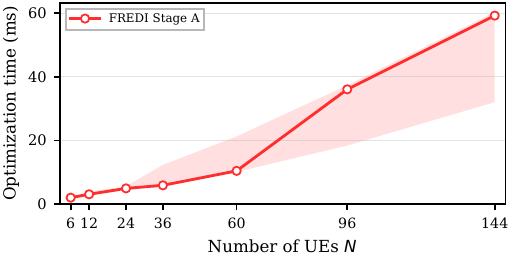}
    \caption{Optimization time of FREDI Stage~A as the number of UEs $N$ increases.}
    \label{fig:s3_scalability}
\end{figure}

%% file: 6__conclusion.tex
\section{Conclusion}
\label{sec:conclusion:TMC}

This paper presented FREDI, a secure cooperative wireless edge-intelligence framework that jointly coordinates dual-threshold inference, UE--ES association, uplink bandwidth and power, and shared ES processing. Stage~A uses proportional fairness, feasible-pair preprocessing, and an exact communication-variable projection to produce a mixed-integer exponential-cone formulation; Stage~B exactly searches the empirical confidence-score domain for each UE. A computable response envelope bounds the loss relative to the empirical global optimum and yields a proportional-response condition for global optimality. Numerically, FREDI maintains near-perfect utility fairness ($J_U\geq0.9993$ on average) while remaining close to the Sum-Utility benchmark; direct optimization of representative $\mathbf P_0$ instances further supports the two-stage decomposition. The results also quantify the effects of security-constrained association and confirm practical Stage-A scalability over the tested system sizes.

%% file: 7__appendices.tex
\appendices

\section{Proof of Lemma~\ref{lem:set_decreasing_monotonic}}
\label{sec:set_decreasing_monotonic}

\begin{proof}
Consider two feasible threshold pairs $(\alpha^l,\alpha^u)$ and $(\widetilde{\alpha}^l,\widetilde{\alpha}^u)$ such that $\widetilde{\alpha}^l\geq\alpha^l$ and $\widetilde{\alpha}^u\geq\alpha^u$. Since the CNN is fixed, the confidence trace ${C_q^{(k)}}_{q=1}^{L}$ of each event $I_k$ is independent of the thresholds.

Suppose that $I_k$ is classified as \emph{critical} under $(\widetilde{\alpha}^l,\widetilde{\alpha}^u)$, first at layer $q^\star$. For every $q<q^\star$, the event has not terminated as normal, so $C_q^{(k)}>\widetilde{\alpha}^l\geq\alpha^l$. If $C_q^{(k)}\geq\alpha^u$ at any such layer, then $I_k$ is already classified as critical under $(\alpha^l,\alpha^u)$; otherwise it remains undecided. At layer $q^\star$, critical classification under the larger thresholds gives $C_{q^\star}^{(k)}\geq\widetilde{\alpha}^u\geq\alpha^u$, so the event is critical under $(\alpha^l,\alpha^u)$ as well. The same argument applies at the final layer, while the equal-threshold convention in~\eqref{eq:label_function} preserves this implication when the undecided interval vanishes.

Hence,
$I_k\in\widehat{\Phi}^{\mathrm{critical}}
(\widetilde{\alpha}^l,\widetilde{\alpha}^u)
\ \Longrightarrow\
I_k\in\widehat{\Phi}^{\mathrm{critical}}
(\alpha^l,\alpha^u),$
and therefore
$\widehat{\Phi}^{\mathrm{critical}}
(\widetilde{\alpha}^l,\widetilde{\alpha}^u)
\subseteq
\widehat{\Phi}^{\mathrm{critical}}
(\alpha^l,\alpha^u)$,
which proves~\eqref{eq:critical_set_monotonicity}. Consequently, $|\widehat{\Phi}^{\mathrm{critical}}|$ is non-increasing in either threshold.
\end{proof}

\section{Proof of Lemma~\ref{lem:exact_communication_projection}}
\label{sec:proof_exact_communication_projection}

\begin{proof}
For a feasible UE--ES pair satisfying $\gamma_{nj}>\gamma_n^{\mathrm{EV}}$ and fixed $b>0$,
\begin{equation}
\frac{\partial r_{nj}^{\mathrm{se}}}{\partial p}
=\frac{b}{\ln 2}
\left(
\frac{\gamma_{nj}}{b+\gamma_{nj}p}
-\frac{\gamma_n^{\mathrm{EV}}}{b+\gamma_n^{\mathrm{EV}}p}
\right)>0,
\end{equation}
so the rate increases in $p$. For fixed $p>0$, write $r_{nj}^{\mathrm{se}}(b,p)=bq_{nj}(p/b)$, where $q_{nj}$ is concave and $q_{nj}(0)=0$. Since $q_{nj}(z)/z$ is non-increasing, $r_{nj}^{\mathrm{se}}(b,p)=p q_{nj}(p/b)/(p/b)$ is non-decreasing in $b$.

For necessity, any feasible selected link satisfies $r_{nj}^{\mathrm{se}}(b_{nj},p^{\max})\geq r_{nj}^{\mathrm{se}}(b_{nj},p_{nj})\geq\bar r_n$ and hence $b_{nj}\geq\underline b_{nj}$. For an unselected link, $b_{nj}\geq\underline b_{nj}x_{nj}$ is immediate. Summing over UEs gives
\begin{equation}
\sum_n\underline b_{nj}x_{nj}\leq\sum_n b_{nj}\leq B_j,
\end{equation}
which proves~\eqref{eq:projected_bandwidth_constraint}. Conversely, if this constraint holds, set $b_{nj}=\underline b_{nj}x_{nj}$ and $p_{nj}=p^{\max}x_{nj}$. Each selected feasible pair meets the secure-rate requirement by the definition of $\underline b_{nj}$; each unselected pair has zero rate and demand. The per-link bounds follow from $\underline b_{nj}\leq b^{\max}$ on $\mathbb M_n^{\mathrm f}$, and the aggregate constraint holds by construction. Thus, the recovered communication allocation is feasible.
\end{proof}

\section{Proof of Theorem~\ref{theo:fredi_global_bound}}
\label{sec:global_performance_bound}

\begin{proof}
For any admissible capacity $\xi$, let $(\alpha_n^{l,\xi},\alpha_n^{u,\xi})$ attain $V_n(\xi)$ and let $L_n^{\xi}$ be its load. Since $0<L_n^{\xi}\leq\xi$,
\[
\kappa_n(\xi)
=\frac{|\Phi_n|V_n(\xi)}{\xi}
\leq
\frac{|\Phi_n|\mathcal U_n(\alpha_n^{l,\xi},\alpha_n^{u,\xi})}{L_n^{\xi}}
\leq\overline{\kappa}_n.
\]
Conversely, evaluating $V_n(\xi)$ at the load of a candidate attaining~\eqref{eq:kappa_upper_envelope} gives the reverse inequality for the supremum; hence,~\eqref{eq:kappa_upper_envelope} is exactly the upper envelope of $\kappa_n(\xi)$ over the positive-utility capacity domain.
For the resource allocation of a global optimizer of~$\mathbf P_0$, replacing any threshold pair by a maximizer in~\eqref{eq:stage_b_value_function} cannot decrease the objective. Its resource variables also satisfy the Stage-A constraints; hence, its processing-coverage term is no larger than the globally optimal Stage-A value $A(\mathbf R^*)$. Applying~\eqref{eq:exact_objective_decomposition} and the preceding response bound therefore gives
\[
F_0^{\star}
\leq
A(\mathbf R^*)+
\sum_{n=1}^{N}\rho_n\ln\overline{\kappa}_n.
\]
The exact Stage-B searches yield $F_0^{\mathrm{FREDI}}=A(\mathbf R^*)+\sum_n\rho_n\ln\kappa_n(R_n^*)$. Subtracting the latter identity from the upper bound proves~\eqref{eq:fredi_global_suboptimality_bound}; the lower inequality follows because the FREDI solution is feasible for~$\mathbf P_0$.
\end{proof}